\documentclass[a4paper, 12pt]{article}
\usepackage{amssymb}
\usepackage{amsmath}
\usepackage{amsthm}
\usepackage[all]{xy}
\usepackage{verbatim}

\newcommand{\be}{\begin{equation}}
\newcommand{\ee}{\end{equation}}
\newcommand{\bea}{\begin{eqnarray}}
\newcommand{\eea}{\end{eqnarray}}

\newcommand{\ri}{\mathrm{i}}

\newcommand{\bN}{\mathbb{N}}
\newcommand{\bR}{\mathbb{R}}
\newcommand{\bC}{\mathbb{C}}

\newcommand{\cP}{\mathcal{P}}

\newcommand{\cO}{\mathcal{O}}
\newcommand{\cV}{\mathcal{V}}
\newcommand{\cF}{\mathcal{F}}
\newcommand{\cA}{\mathcal{A}}
\newcommand{\cM}{\mathcal{M}}
\newcommand{\cS}{\mathcal{S}}

\newcommand{\mfa}{\mathfrak{a}}
\newcommand{\mfc}{\mathfrak{c}}
\newcommand{\mfm}{\mathfrak{m}}
\newcommand{\mfu}{\mathfrak{u}}
\newcommand{\mfg}{\mathfrak{g}}
\newcommand{\mfgl}{\mathfrak{gl}}
\newcommand{\mfk}{\mathfrak{k}}
\newcommand{\mfp}{\mathfrak{p}}
\newcommand{\mfX}{\mathfrak{X}}
\newcommand{\mfs}{\mathfrak{s}}

\newcommand{\mfL}{\mathfrak{L}}

\newcommand{\bsone}{\boldsymbol{1}}
\newcommand{\bsq}{\boldsymbol{q}}

\newcommand{\bslambda}{\boldsymbol{\lambda}}

\newcommand{\bsX}{\boldsymbol{X}}

\newcommand{\bsV}{\boldsymbol{V}}
\newcommand{\bss}{\boldsymbol{s}}

\newcommand{\diag}{\text{diag}}

\newcommand{\pr}{{\text{pr}}}

\newcommand{\ad}{\mathrm{ad}}

\newcommand{\Ad}{\mathrm{Ad}}
\newcommand{\tr}{\mathrm{tr}}
\newcommand{\Id}{\mathrm{Id}}

\newcommand{\bc}{\text{BC}}

\newcommand{\ddd}{\mathrm{d}}
\newcommand{\ext}{\mathrm{ext}}

\newcommand{\acts}{\, . \,}
\newcommand{\inner}{\, \lrcorner \,}

\newcommand{\half}{\frac{1}{2}}

\theoremstyle{plain}
\newtheorem{THEOREM}{Theorem}
\newtheorem{LEMMA}[THEOREM]{Lemma}

\begin{document}
\begin{center}
\Large{\textbf{
The hyperbolic $BC_n$ Sutherland and the rational 
$BC_n$ Ruijsenaars--Schneider--van Diejen models:
Lax matrices and duality
}}
\end{center}
\bigskip
\begin{center}
B.G.~Pusztai\\
Bolyai Institute, University of Szeged,\\
Aradi v\'ertan\'uk tere 1, H-6720 Szeged, 
Hungary\\
e-mail: \texttt{gpusztai@math.u-szeged.hu}
\end{center}
\bigskip
\begin{abstract}
In this paper, we construct canonical action-angle 
variables for both the hyperbolic $BC_n$ Sutherland 
and the rational $BC_n$ 
Ruijsenaars--Schneider--van Diejen models with 
\emph{three independent coupling constants}.
As a byproduct of our symplectic reduction approach, 
we establish the action-angle duality between these
many-particle systems. The presented dual reduction 
picture builds upon the construction of a Lax 
matrix for the $BC_n$-type rational 
Ruijsenaars--Schneider--van Diejen model.



\end{abstract}
\newpage

\section{Introduction}
\label{S1}
\setcounter{equation}{0}
The Calogero--Moser--Sutherland (CMS) and the 
Ruijsenaars--Schneider--van Diejen (RSvD) 
interacting many-particle models play a 
distinguished role among the integrable 
Hamiltonian systems, having numerous 
relationships to important fields of 
mathematics and physics. They have profound 
applications in the theory of solitons 
(see e.g. \cite{RuijSchneider},
\cite{RuijFiniteDimSolitonSystems},
\cite{BabelonBernard}, 
\cite{KapustinSkorik}), 
and recently they appeared in the context of 
random matrix theory as well 
(see e.g.
\cite{Bogomolny_PRL}, 
\cite{Bogomolny}).
Quite surprisingly, these intriguing relationships 
are well-understood only for the models associated
with the $A_n$ root system. It appears that the 
main technical obstacle for developing analogous 
theories in association with the non-$A_n$-type 
root systems is the lack of knowledge of explicit 
action-angle variables for the non-$A_n$-type CMS 
and RSvD models. In this paper we wish to 
narrow this gap by constructing action-angle 
systems of canonical coordinates for both the 
hyperbolic $BC_n$ Sutherland and the rational 
$BC_n$ RSvD models. 

In order to define the $BC_n$-type hyperbolic 
Sutherland and the rational RSvD many-particle 
systems, we first introduce the subset
\be
\mfc = \{ x = (x_1, \ldots, x_n) \in \bR^n
\, | \,
x_1 > \ldots > x_n > 0 \}
\subset \bR^n,
\label{mfc}
\ee
which can be seen as an appropriate model
for the open Weyl chamber of type $BC_n$.
Let us recall that the phase space of the
Sutherland model is the cotangent bundle of 
$\mfc$, for which we have the natural 
identification
\be
\cP^S = \mfc \times \bR^n
= \{ (q, p) 
\, | \, 
q \in \mfc \, , p \in \bR^n \}.
\label{cP_S}
\ee
Recall also that the hyperbolic $BC_n$ 
Sutherland dynamics is generated by the 
interacting many-particle Hamiltonian
\be
\begin{split}
H^S
= & \half \sum_{c = 1}^n p_c^2
+ \sum_{1 \leq a < b \leq n}
\left(
\frac{g^2}{\sinh^2(q_a - q_b)}
+ \frac{g^2}{\sinh^2(q_a + q_b)}
\right) 
\\
& + \sum_{c = 1}^n \frac{g_1^2}{\sinh^2(q_c)}
+ \sum_{c = 1}^n \frac{g_2^2}{\sinh^2(2 q_c)},
\label{H_S}
\end{split}
\ee
where the so-called coupling parameters
$g$, $g_1$ and $g_2$ are arbitrary real
numbers satisfying the inequalities $g^2 > 0$ 
and $g_1^2 + g_2^2 > 0$. In other words, we are 
interested only in the hyperbolic $BC_n$
Sutherland model\footnote{Notice that with
the specialization $g_2 = 0$ we recover
the hyperbolic $B_n$ Sutherland particle system,
meanwhile with $g_1 = 0$ we obtain the
$C_n$-type model.} with purely \emph{repulsive} 
interaction. Since the strength of the interaction 
is governed by the numbers
$g^2$, $g_1^2$ and $g_2^2$, they are usually 
called the coupling constants.  

Though the non-$A_n$-type CMS models
have received a lot of attention in the
last couple of decades (see e.g. the 
fundamental papers \cite{OlshaPere76}, 
\cite{OlshaPere} and the book
\cite{PerelomovBook}), the symplectic
reduction understanding of the $BC_n$ Sutherland
model with three independent coupling constants
is a quite recent development 
\cite{FeherPusztai2007}. 
Besides providing a Lax representation of the 
dynamics, the symplectic reduction approach has
the advantage that it naturally leads to a 
fairly simple solution algorithm of purely
algebraic nature. By pushing forward the reduction 
picture, in this paper we are able to furnish 
action-angle variables for the standard
$BC_n$ Sutherland model with repulsive 
interaction.

The non-$A_n$-type deformations of the classical 
Ruijsenaars--Schneider many-particle systems have 
been introduced by van Diejen \cite{vanDiejen1994}. 
Just as for the Sutherland model, the phase 
space of the rational $BC_n$ RSvD model is the 
cotangent bundle $T^* \mfc$, which is naturally 
identified with the manifold
\be
\cP^R = \mfc \times \bR^n
= \{ (\lambda, \theta) 
\, | \, 
\lambda \in \mfc \, , \theta \in \bR^n \}.
\label{cP_R}
\ee
Recall that the rational RSvD dynamics is 
characterized by the Hamiltonian
\be
\begin{split}
H^R 
= & \sum_{c = 1}^{n} \cosh(2 \theta_c)
\left( 
1 + \frac{\nu^2}
{\lambda_c^2} 
\right)^\frac{1}{2}
\left( 
1 + \frac{\kappa^2}
{\lambda_c^2} 
\right)^\frac{1}{2}
\prod_{\substack{d = 1 \\ (d \neq c)}}^{n}
\left( 1 
+ \frac{4 \mu^2}
{(\lambda_c - \lambda_d)^2} \right)^\frac{1}{2}
\left( 1 
+ \frac{4 \mu^2}
{(\lambda_c + \lambda_d)^2} \right)^\frac{1}{2}
\\
& + \frac{\nu \kappa}{4 \mu^2}
\prod_{c = 1}^n 
\left(
1 + \frac{4 \mu^2}{\lambda_c^2}
\right)
- \frac{\nu \kappa}{4 \mu^2},
\label{H_R}
\end{split}
\ee
where $\mu$, $\nu$ and $\kappa$ are arbitrary
real parameters satisfying $\mu \neq 0 \neq \nu$. 
Although the Liouville integrability of the 
non-$A_n$-type RSvD models 
has been verified \cite{vanDiejen1994}, the
Lax representation of their dynamics is
still missing. However, by generalizing our 
results on the $C_n$-type RSvD model 
\cite{Pusztai_NPB2011}, in this paper we provide 
a Lax matrix and an elementary solution algorithm 
for the rational RSvD model (\ref{H_R})
with $\nu \kappa \geq 0$. Moreover, the 
proposed reduction approach permits us to 
construct action-angle variables as well.

The organization of the paper can be outlined
as follows. Section \ref{S2} is devoted to a 
brief account on the necessary group theoretic 
and symplectic geometric background underlying 
the derivation of the Sutherland and the RSvD 
models from a unified symplectic reduction 
framework. In Section \ref{S3} we review the 
symplectic reduction understanding of the $BC_n$ 
Sutherland model. Although this is a standard 
material (see \cite{FeherPusztai2007}),
our new contribution on the spectral properties
of the Lax matrix of the Sutherland model, 
formulated in Lemma \ref{L1}, seems to be crucial 
in advancing the reduction approach to cover the 
RSvD model, too. Starting with Section \ref{S4} 
we present our new results on the rational $BC_n$
RSvD model. By fitting the $BC_n$ RSvD model
into a convenient symplectic reduction picture, 
we are able to provide a Lax matrix and an
elementary solution algorithm as well. The main 
technical result is Theorem \ref{T5},
in which we confirm that the parametrization
of the Lax matrix of the RSvD model does provide 
a Darboux system on the reduced phase space. In 
Section \ref{S5} we elaborate on the consequences 
of the proposed reduction approach. In particular, 
a natural construction of canonical action-angle 
variables for both the Sutherland and the RSvD 
model comes for free. Furthermore, the action-angle 
duality between the repulsive $BC_n$ Sutherland 
model and the $BC_n$ RSvD system with 
$\nu \kappa \geq 0$ becomes also transparent.

This paper is a continuation of our recent work 
\cite{Pusztai_NPB2011} on the hyperbolic Sutherland 
and the rational RSvD models associated with the 
$C_n$ root system. It is a very fortunate situation 
that many results for the models associated with 
the $BC_n$ root system can be derived almost 
effortlessly by generalizing the analogous results 
of the $C_n$-type particle systems. Since in 
\cite{Pusztai_NPB2011} we have carried out a very 
detailed analysis on the particle systems of
type $C_n$, in this paper we can be brief on many 
aspects of the $BC_n$-type models. Though our 
presentation tries to be self-contained, in this 
paper we rather focus on differences between the 
$C_n$-type and the $BC_n$-type models, and we 
provide proofs only for those facts that have no 
natural analogs in the $C_n$ case. Therefore the 
reader may find it useful to have a copy of 
\cite{Pusztai_NPB2011} on hand while reading 
the paper.

\section{Preliminaries}
\label{S2}
\setcounter{equation}{0}
In this section we gather the necessary group
theoretic and symplectic geometric material 
underlying the unified symplectic reduction 
derivation of the hyperbolic $BC_n$ Sutherland 
and the rational $BC_n$ RSvD models. Throughout 
the paper our group theoretic conventions try 
to be consistent with the book \cite{Knapp}, 
whereas the symplectic geometric conventions 
come mainly from \cite{AM}. To facilitate the 
comparison with our work on the $C_n$-type models, 
the majority of the notations are directly 
borrowed from paper \cite{Pusztai_NPB2011}. 

Take an arbitrary positive integer
$n \in \bN = \{1, 2, \ldots \}$ and let $N = 2 n$. 
With the aid of the $N \times N$ unitary matrix
\be
C = \begin{bmatrix}
0_n & \bsone_n \\
\bsone_n & 0_n
\end{bmatrix} \in U(N)
\label{C}
\ee
we define the non-compact real reductive matrix
Lie group
\be
G = U(n, n) 
= \{ y \in GL(N, \bC) 
\, | \, 
y^* C y = C \}.
\label{G}
\ee
The corresponding real matrix Lie algebra has 
the form
\be
\mfg = \mfu(n, n)
= \{ Y \in \mfgl(N, \bC) 
\, | \, 
Y^* C + C Y = 0 \},  
\label{mfg}
\ee
on which the map
\be
\langle \, , \rangle \colon 
\mfg \times \mfg \rightarrow \bR,
\quad
(Y_1, Y_2) \mapsto 
\langle Y_1, Y_2 \rangle = \tr(Y_1 Y_2)
\label{bilinear_form}
\ee
provides a symmetric $\Ad$-invariant non-degenerate 
bilinear form. 

Let us remember that the fixed-point set of the 
Cartan involution $\Theta (y) = (y^*)^{-1}$ 
$(y \in G)$ naturally selects a maximal compact 
subgroup
\be
K = \{ y \in G 
\, | \, 
\Theta (y) = y \}
= \{ y \in G 
\, | \, 
y \mbox{ is unitary} \}
\cong U(n) \times U(n)
\label{K}
\ee
of the Lie group $G$. Also, the Lie algebra 
involution $\theta (Y) = - Y^*$ $(Y \in \mfg)$ 
corresponding to $\Theta$ naturally induces the 
Cartan decomposition
\be
\mfg = \mfk \oplus \mfp
\label{gradation}
\ee
with the Lie subalgebra and the complementary
subspace
\be
\mfk = \ker(\theta - \Id) 
= \{ Y \in \mfg \, | \, Y^* = - Y \}
\quad \mbox{and} \quad
\mfp = \ker(\theta + \Id) 
= \{ Y \in \mfg \, | \, Y^* = Y \},
\label{mfkp}
\ee
respectively. Due to the Cartan decomposition 
(\ref{gradation}), each $Y \in \mfg$ can be 
decomposed uniquely as
\be
Y = Y_+ + Y_-
\qquad
(Y_+ \in \mfk, \, Y_- \in \mfp).
\label{decomp}
\ee

Next, notice that the set of diagonal matrices
\be
\mfa = \{ 
Q = \diag(q_1, \ldots, q_n, -q_1, \ldots, -q_n)
\in \mfp 
\, | \,
(q_1, \ldots, q_n) \in \bR^n \}
\label{mfa}
\ee
forms a maximal Abelian subspace in $\mfp$ 
(\ref{mfkp}). Let $\mfa^\perp$ denote the 
subspace of the off-diagonal elements of $\mfp$; 
then we have the orthogonal decomposition 
$\mfp = \mfa \oplus \mfa^\perp$. Let us also
consider the centralizer of $\mfa$ inside $K$,
which is the Abelian subgroup
\be
M = Z_K (\mfa)
= \{
\diag( e^{\ri \chi_1}, \ldots, e^{\ri \chi_n},
e^{\ri \chi_1}, \ldots, e^{\ri \chi_n} ) \in K
\, | \,
(\chi_1, \ldots, \chi_n) \in \bR^n \}.
\label{M}
\ee
Obviously its Lie algebra has the form
\be
\mfm = \{
\diag(\ri \chi_1, \ldots, \ri \chi_n, 
\ri \chi_1, \ldots, \ri \chi_n)
\in \mfk 
\, | \,
(\chi_1, \ldots, \chi_n) \in \bR^n \}.
\label{mfm} 
\ee
If $\mfm^\perp$ denotes the subspace of the
off-diagonal elements of $\mfk$, then we can
write $\mfk = \mfm \oplus \mfm^\perp$. Finally,
recalling the Cartan decomposition 
(\ref{gradation}), we end up with the refined
decomposition
\be
\mfg 
= \mfm \oplus \mfm^\perp 
\oplus \mfa \oplus \mfa^\perp.
\label{refined_decomp}
\ee

Having equipped with the above group theoretic
objects, in the rest of the section we review
some basic notions from symplectic geometry.
Recall that the cotangent bundle $T^* G$ of the 
Lie group $G$ can be trivialized, say, by 
left translations. Upon identifying the dual 
space $\mfg^*$ with the Lie algebra $\mfg$ via 
the bilinear form (\ref{bilinear_form}), it is 
clear that the product manifold
\be
\cP = G \times \mfg
= \{ (y, Y) 
\, | \, 
y \in G, \, Y \in \mfg \}
\label{cP}
\ee
provides a convenient model for $T^* G$. 
Furthermore, the tangent spaces of $\cP$ can be 
naturally identified as
\be
T_{(y, Y)} \cP = T_{(y, Y)}(G \times \mfg) 
\cong T_y G \oplus T_Y \mfg 
\cong T_y G \oplus \mfg
\qquad
((y, Y) \in \cP).
\label{tangent_space_identification}
\ee
Let us observe that on the model space 
$\cP \cong T^* G$ the canonical one-form 
$\vartheta \in \Omega^1(\cP)$ reads
\be
\vartheta_{(y, Y)} (\delta y \oplus \delta Y) 
= \langle y^{-1} \delta y, Y \rangle
\qquad
((y, Y) \in \cP, \, 
\delta y \oplus \delta Y \in T_y G \oplus \mfg ),
\label{canonical_theta}
\ee
whereas for the canonical symplectic form we use 
the convention 
$\omega = -\ddd \vartheta \in \Omega^2(\cP)$.
 
Now note that the smooth left action of the 
product Lie group $K \times K$ on the group 
manifold $G$ defined by the formula
\be
(k_L, k_R) \acts y = k_L y k_R^{-1}
\qquad
(y \in G, \, (k_L, k_R) \in K \times K)
\label{KK_action_on_G}
\ee
naturally lifts onto $T^* G$. 
Working with the model space $\cP$ (\ref{cP})
of the cotangent bundle, 
the lift of the above $K \times K$-action 
(\ref{KK_action_on_G}) takes the form
\be
(k_L, k_R) \acts (y, Y) 
= (k_L y k_R^{-1}, k_R Y k_R^{-1})
\qquad
((y, Y) \in \cP, 
\, (k_L, k_R) \in K \times K).
\label{KKonP}
\ee
This action is clearly symplectic, admitting the 
$K \times K$-equivariant momentum map
\be
J \colon \cP 
\rightarrow 
(\mfk \oplus \mfk)^* \cong \mfk \oplus \mfk,
\quad (y, Y) \mapsto J(y, Y) 
= (y Y y^{-1})_+ \oplus (-Y_+).
\label{J}
\ee
Without any further notice, in the rest of the 
paper we shall frequently use the natural dual 
space identification
$(\mfk \oplus \mfk)^* \cong \mfk \oplus \mfk$
induced by the bilinear form 
(\ref{bilinear_form}).

To proceed further, with each column vector 
$V \in \bC^N$ subject to the conditions 
$V^* V = N$ and $C V + V = 0$ we associate 
the Lie algebra element
\be
\xi(V) 
= \ri \mu (V V^* - \bsone_N) + \ri (\mu - \nu) C 
\in \mfk,
\label{xi}
\ee
where $\mu$, $\nu$ are arbitrary real parameters
satisfying $\mu \neq 0 \neq \nu$. Also, let 
$E \in \bC^N$ denote the distinguished column 
vector with components
\be
E_a = - E_{n + a} = 1 \qquad 
(a \in \bN_n = \{1, \ldots, n \}),
\label{E}
\ee
and consider the Lie algebra element
\be
J_0 = (-\xi(E)) \oplus \ri \kappa C
\in \mfk \oplus \mfk,
\label{J_0}
\ee
where $\kappa$ is an arbitrary real parameter.
In order to derive the hyperbolic $BC_n$ Sutherland 
and the rational $BC_n$ RSvD models from symplectic
reduction, we wish to 
reduce the symplectic manifold $(\cP, \omega)$ 
at the very special value $J_0$ (\ref{J_0}) of 
the momentum map $J$ (\ref{J}). We mention in 
passing that the parametrizations of the Lie 
algebra elements $\xi(V)$ (\ref{xi}) and $J_0$ 
(\ref{J_0}) turn out to be very natural in the 
sense that, after performing the reduction, the 
parameter triple $(\mu, \nu, \kappa)$ can be 
identified with the coupling parameters of the 
rational $BC_n$ RSvD model (\ref{H_R}).

Our experience with the CMS and the RSvD models 
convinces us that, in general, the application 
of the shifting trick leads to a shorter and 
neater derivation of these particle systems from
symplectic reduction. As the initial step
of the shifting trick (see e.g. \cite{OR}),
we have to identify the adjoint orbit passing 
through $-J_0$ (\ref{J_0}). Since $C$ commutes 
with each element of $K$, for the adjoint 
orbit in question we have the natural 
identification 
$\cO \oplus \{- \ri \kappa C \} \cong \cO$, where
\be
\cO = \cO(\xi(E)) 
= \{ \xi(V) \in \mfk \, | \, V \in \bC^N, 
\,
V^* V = N, 
\, C V + V = 0 \}.
\label{cO}
\ee
Following the prescription of the shifting trick,
we also introduce the extended phase space
\be
\cP^\ext = \cP \times \cO
= \{ (y, Y, \rho)
\, | \,
y \in G, \, Y \in \mfg \, ,\rho \in \cO \},
\label{cP_ext}
\ee
and endow it with the product symplectic structure
\be
\omega^\ext = \omega + \omega^\cO,
\label{omega_ext}
\ee 
where, of course, $\omega^\cO$ is the standard 
Kirillov--Kostant--Souriau symplectic form carried 
by the orbit $\cO$ (\ref{cO}). The natural
extension of the $K \times K$-action (\ref{KKonP}) 
onto $\cP^\ext$ is given by the diagonal action
\be
(k_L, k_R) \acts (y, Y, \rho) 
= (k_L y k_R^{-1}, k_R Y k_R^{-1}, 
k_L \rho k_L^{-1}),
\label{KKonPext}
\ee
and the corresponding $K \times K$-equivariant 
momentum map takes the form
\be
J^\ext \colon 
\cP^\ext \rightarrow \mfk \oplus \mfk,
\quad
(y, Y, \rho)
\mapsto
J^\ext(y, Y, \rho) 
= ( (y Y y^{-1})_+ + \rho ) 
\oplus (- Y_+ -\ri \kappa C).
\label{J_ext}
\ee
As a matter of fact, it is clear that $J^\ext$
takes its values in the subalgebra
\be
\mfs (\mfk \oplus \mfk)
= \{ X_L \oplus X_R \in \mfk \oplus \mfk
\, | \,
\tr(X_L) + \tr(X_R) = 0 \} 
\leq \mfk \oplus \mfk.
\label{mfs_mfk_mfk}
\ee
Now, the shifting trick guarantees that 
\be
\cP /\!/_{J_0} (K \times K)
\cong
\cP^\ext /\!/_{0} (K \times K),
\ee
i.e. for our purposes it is an equally 
valid approach to perform the Marsden--Weinstein
reduction of the symplectic manifold 
$(\cP^\ext, \omega^\ext)$ 
at the zero value of the momentum map $J^\ext$ 
(\ref{J_ext}). 

\section{The hyperbolic $BC_n$ Sutherland model}
\label{S3}
\setcounter{equation}{0}
In this section we review the Lax matrix and the
symplectic reduction understanding of the hyperbolic 
$BC_n$ Sutherland model with three independent 
coupling constants. Our discussion on the reduction 
aspects of the model is mainly based on 
the ideas presented in \cite{FeherPusztai2007}, 
adapted to the conventions of \cite{Pusztai_NPB2011}.  

\subsection{The Lax matrix of the Sutherland model}
The main goal of this subsection is to introduce 
the Lax matrix of the Sutherland model, and to 
analyze some of its spectral properties that prove 
to be pertinent in the symplectic geometric 
understanding of the $BC_n$ RSvD model. As a 
preparatory step, with each point $(q, p) \in \cP^S$ 
we associate the diagonal matrices
\be
Q = \diag(q_1, \ldots, q_n, -q_1, \ldots, -q_n) 
\in \mfa
\quad \mbox{and} \quad
P = \diag(p_1, \ldots, p_n, -p_1, \ldots, -p_n)
\in \mfa. 
\ee
Let $\tilde{\ad}_Q$ denote the restriction of the 
linear operator $\ad_Q = [Q, \cdot] \in \mfgl(\mfg)$ 
onto the off-diagonal part of the Lie algebra $\mfg$
(\ref{mfg}). Notice that the regularity condition 
$q \in \mfc$ ensures the invertibility of the linear 
operator $\tilde{\ad}_Q$. Therefore, making use of 
the standard functional calculus, the matrix 
\be 
L_\mfp(q, p) = P - \sinh(\tilde{\ad}_Q)^{-1} \xi(E)
+ \coth(\tilde{\ad}_Q)(\ri \kappa C) \in \mfp 
\label{L_mfp}
\ee
is well-defined. Let us note that the above 
introduced $N \times N$ matrix 
$L_\mfp = L_\mfp(q, p)$ is Hermitian with 
block-matrix structure
\be
L_\mfp = \begin{bmatrix}
A & B \\
-B & -A
\end{bmatrix},
\label{L_mfp_blocks}
\ee
where $A$ and $B$ are $n \times n$ matrices 
satisfying $A^* = A$ and $B^* = -B$. More
concretely, for their matrix entries we have
\be
A_{a, b} = \frac{-\ri \mu}{\sinh(q_a - q_b)}, 
\quad
A_{c, c} = p_c, 
\quad
B_{a, b} = \frac{\ri \mu}{\sinh(q_a + q_b)}, 
\quad
B_{c, c} 
= \frac{\ri \nu}{\sinh(2 q_c)} 
+ \ri \kappa \coth(2 q_c), 
\label{A&B_entries} 
\ee
where $a, b, c \in \bN_n$ and $a \neq b$. Now, 
with the parameter triple $(\mu, \nu, \kappa)$
we associate the map 
\be
L \colon \cP^S \rightarrow \mfg,
\quad
(q, p) \mapsto L(q, p)
= L_\mfp(q, p) - \ri \kappa C.
\label{L}
\ee
As one can see in \cite{FeherPusztai2007}, the 
above map $L$ provides a Lax matrix for the 
hyperbolic $BC_n$ Sutherland model. The exact 
relationship between the parameters 
$(\mu, \nu, \kappa)$ and the Sutherland
coupling parameters $(g, g_1, g_2)$ appearing
in (\ref{H_S}) will be clarified later (see 
(\ref{S_coupling_parameters_from_RSvD})
and (\ref{RSvD_coupling_parameters_from_S})).

Having defined the Lax matrix of the $BC_n$
Sutherland model, we now turn our
attention to its spectral properties. Remembering
that the only difference between $L$ and
$L_\mfp$ is the anti-Hermitian constant term 
$\ri \kappa C$, it is clear that the 
spectral properties of the \emph{non-Hermitian} 
Lax matrix $L$ (\ref{L}) can be understood by 
analyzing the spectrum of the Hermitian 
matrix $L_\mfp$ (\ref{L_mfp}). Since $L_\mfp$
belongs to the complementary subspace $\mfp$
(\ref{mfkp}), we know from general principles 
that it can be conjugated into the maximal Abelian 
subspace $\mfa$ (\ref{mfa}) by some element of the 
maximal compact subgroup $K$ (\ref{K}). However, 
this diagonalization procedure becomes much more 
explicit by exploiting the 
\emph{singular value decomposition} 
of the sum of the matrices $A$ and $B$ introduced 
in the block-matrix decomposition 
(\ref{L_mfp_blocks}). More precisely, we can write
\be
A + B = u \bss v^*,
\label{SVD}
\ee
where $u$ and $v$ are $n \times n$ 
unitary matrices, meanwhile 
$\bss = \diag(s_1, \ldots, s_n)$
is a diagonal matrix filled in with the singular
values $s_1 \geq \ldots \geq s_n \geq 0$ of the 
matrix $A + B$. Now, upon defining the $N \times N$ 
block-matrices
\be
k = \half \begin{bmatrix}
v + u & v - u \\
v - u & v + u
\end{bmatrix}
\quad \mbox{and} \quad
S = \begin{bmatrix}
\bss & 0 \\
0 & -\bss
\end{bmatrix},
\ee
one can easily verify that $k \in K$, $S \in \mfa$ 
and $L_\mfp = k S k^{-1}$. Having diagonalized the
matrix $L_\mfp$ (\ref{L_mfp}), from the definition 
(\ref{L}) we see at once that
\be
L^2 = k (S^2 - \kappa^2 \bsone_N) k^{-1},
\label{L_squared}
\ee
therefore the spectrum of the Hermitian matrix
$L^2$ can be identified as
\be
\sigma(L^2) 
= \sigma(S^2 - \kappa^2 \bsone_N)
= \sigma(\bss^2 - \kappa^2 \bsone_n).
\label{L_squared_spectrum_1}
\ee 
On the other hand, remembering the singular value 
decomposition (\ref{SVD}), we can also write
\be
(A + B) (A - B) 
= (A + B) (A + B)^*
= u \bss v^* v \bss u^* 
= u \bss^2 u^{-1},
\ee
from where we obtain the spectral identification
\be
\sigma(\bss^2) = \sigma(A^2 - B^2 - [A, B]).
\label{bss&AB}
\ee
Now, the comparison of the equations 
(\ref{L_squared_spectrum_1}) and
(\ref{bss&AB}) immediately leads to
the formula
\be
\sigma(L^2) 
= \sigma(A^2 - B^2 - [A, B] - \kappa^2 \bsone_n).
\label{L_squared_spectrum_OK}
\ee
Since $L^2$ is an Hermitian matrix, the spectral
mapping theorem guarantees that each 
eigenvalue of $L$ is either a real number or 
a purely imaginary number. However, under certain
technical assumptions, the relationship
(\ref{L_squared_spectrum_OK}) permits us to
provide a more accurate description for the
spectrum of $L$.

\begin{LEMMA} 
\label{L1}
Suppose that $\nu \neq 2 \mu$ and 
$\nu \kappa \geq 0$; then for each point 
$(q, p) \in \cP^S$ we have $L(q, p)^2 > 0$, 
i.e. the matrix $L(q, p)^2$ is positive definite. 
In particular, the eigenvalues of the Lax matrix 
$L(q, p)$ are non-zero real numbers.
\end{LEMMA}

\begin{proof}
Take an arbitrary point $(q, p) \in \cP^S$
and keep it fixed. First, notice that if 
$\kappa = 0$, then the Lax matrix $L = L(q, p)$ 
is of type $C_n$. However, we have a
fairly complete knowledge on the spectrum
of the Lax matrix of the $C_n$ Sutherland model. 
Namely, since $\nu \neq 2 \mu$, from Lemma 1 
in \cite{Pusztai_JPA2011} we see that the 
Hermitian matrix $L$ is invertible, whence 
$L^2 > 0$ is immediate.

In the following we assume that $\kappa \neq 0$. 
As an important auxiliary object in our proof, 
let us consider the $BC_n$-type Lax matrix 
$\check{L} = \check{L}(q, p)$ associated with 
the parameters $(\mu, \nu - \kappa, 0)$.
Recalling (\ref{L}) and (\ref{L_mfp_blocks}),
we see that $\check{L}$ is an Hermitian matrix
with block-matrix decomposition
\be
\check{L} = \check{L}_\mfp
= \begin{bmatrix}
\check{A} & \check{B} \\
-\check{B} & -\check{A}
\end{bmatrix} \in \mfp.
\label{L_check}
\ee
Moreover, remembering (\ref{A&B_entries}), for 
the matrix entries of $\check{A}$ and $\check{B}$ 
we have 
\be
\check{A}_{a, b} 
= \frac{-\ri \mu}{\sinh(q_a - q_b)}, 
\quad
\check{A}_{c, c} = p_c, 
\quad
\check{B}_{a, b} 
= \frac{\ri \mu}{\sinh(q_a + q_b)}, 
\quad
\check{B}_{c, c} 
= \frac{\ri (\nu - \kappa)}{\sinh(2 q_c)},  
\label{Acheck&Bcheck_entries} 
\ee
where $a, b, c \in \bN_n$ and $a \neq b$.
Since $\check{L}$ is Hermitian, it is clear 
that $\check{L}^2 \geq 0$. Thus the direct 
application of (\ref{L_squared_spectrum_OK}) 
on the Lax matrix $\check{L}$ yields immediately 
that 
\be
\check{A}^2 - \check{B}^2 
- [\check{A}, \check{B}] 
\geq 0.
\label{checked_positivity}
\ee
In order to find the connection between the Lax 
matrix $L$ (\ref{L}) and the auxiliary Lax matrix 
$\check{L}$ (\ref{L_check}), we introduce the 
diagonal matrix $\bsq = \diag(q_1, \ldots, q_n)$.  
Comparing the equations (\ref{A&B_entries}) and 
(\ref{Acheck&Bcheck_entries}), it is obvious that 
\be
A = \check{A}
\quad \mbox{and} \quad 
B = \check{B} + \ri \kappa \coth(\bsq).
\ee
Thus it is immediate that
\begin{multline}
A^2 - B^2 - [A, B] - \kappa^2 \bsone_N
\\
= \check{A}^2 - \check{B}^2 
- [\check{A}, \check{B}]
+ \ri \kappa [\coth(\bsq), \check{A}]
- \ri \kappa (\coth(\bsq) \check{B}
+ \check{B} \coth(\bsq))
+ \kappa^2 \sinh(\bsq)^{-2}.
\label{lemma1_key}
\end{multline}
Upon introducing the
column vector $\bsV \in \bC^n$ with components
$\bsV_c = 1 / \sinh(q_c)$ $(c \in \bN_n)$,
the right hand side of the above equation
can be simplified considerably. Indeed, 
by applying the standard hyperbolic identity
\be
\frac{\coth(x) + \coth(y)}{\sinh(x + y)}
= \frac{1}{\sinh(x) \sinh(y)},
\label{func_identity}
\ee 
one can easily verify the relations
\begin{align}
& \ri \kappa ( \coth(\bsq) \check{A}
- \check{A} \coth(\bsq)) 
= - \mu \kappa \bsV \bsV^* 
+ \mu \kappa \sinh(\bsq)^{-2}, \\
& \ri \kappa ( \coth(\bsq) \check{B}
+ \check{B} \coth(\bsq))
= - \mu \kappa \bsV \bsV^*
+ (\mu \kappa - \nu \kappa + \kappa^2)
\sinh(\bsq)^{-2}.
\end{align}
Plugging these formulae into (\ref{lemma1_key}),
we end up with the concise expression
\be
A^2 - B^2 - [A, B] - \kappa^2 \bsone_N
= \check{A}^2 - \check{B}^2
- [\check{A}, \check{B}]
+ \nu \kappa \sinh(\bsq)^{-2}.
\ee
However, due to (\ref{checked_positivity}) 
and the assumption $\nu \kappa \geq 0$,
the matrix on the right hand side of the 
above equation is manifestly positive definite. 
Therefore, according to the spectral
identification (\ref{L_squared_spectrum_OK}), 
we conclude that $L^2 > 0$.
\end{proof}

\subsection{The phase space of the 
Sutherland model}
In this subsection we perform the Marsden--Weinstein 
reduction of the extended symplectic manifold 
$(\cP^\ext, \omega^\ext)$ at the zero value of the 
momentum map $J^\ext$ (\ref{J_ext}). As a first step 
of the reduction, we have to solve the constraint
\be
J^\ext(y, Y, \rho) = 0
\label{constraint}
\ee
for $(y, Y, \rho) \in \cP^\ext$. In other words, we
have to understand the differential geometric 
properties of the closed level set
\be
\mfL_0 = (J^\ext)^{-1}(\{ 0 \})
= \{ (y, Y, \rho) \in \cP^\ext 
\, | \,
J^\ext(y, Y, \rho) = 0 \} 
\subset \cP^\ext.
\label{mfL_0}
\ee
In order to derive the phase space of the Sutherland 
model from the proposed reduction picture, we are 
looking for a special parametrization of $\mfL_0$ 
induced by the \emph{$KAK$ decomposition} of the 
group elements $y \in G$. (For background 
information on the $KAK$ decomposition see e.g. 
the book \cite{Knapp}.) As it can be seen from the 
lemma below, besides the $KAK$ decomposition, the 
most important ingredient of the parametrization 
is the Lax operator (\ref{L}).

\begin{LEMMA}
\label{L2}
Suppose that $\nu + \kappa \neq 0$; then for each 
point $(y, Y, \rho) \in \mfL_0$ of the level set 
there are some $q \in \mfc$, $p \in \bR^n$ and 
$\eta_L, \eta_R \in K$, such that
\be
y = \eta_L e^Q \eta_R^{-1},
\quad
Y = \eta_R L(q, p) \eta_R^{-1},
\quad
\rho = \eta_L \xi(E) \eta_L^{-1}.
\ee
\end{LEMMA}

\begin{proof}
Take an arbitrary point $(y, Y, \rho) \in \mfL_0$.
The $KAK$ decomposition tells us precisely that the
Lie group element $y \in G$ can be decomposed as 
\be
y = k_L e^Q k_R^{-1},
\ee
where $k_L, k_R \in K$ and 
\be
Q = \diag(q_1, \ldots, q_n, -q_1, \ldots, -q_n) 
\in \mfa
\ee 
with some $q_1 \geq \ldots \geq q_n \geq 0$. Also,
by (\ref{cO}),
we can write $\rho = \xi(V)$ with some column vector 
$V \in \bC^N$ satisfying $V^* V = N$ and 
$C V + V = 0$. 

Plugging the above parametrizations into the 
constraint $J^\ext(y, Y, \rho) = 0$,
the explicit form of the momentum map $J^\ext$
(\ref{J_ext}) immediately leads to the relationship 
$Y_+ = - \ri \kappa C$, together with
\be
0 = (y Y y^{-1})_+ + \rho 
= k_L \left(
\sinh(\ad_Q) (k_R^{-1} Y_- k_R)
- \cosh(\ad_Q)(\ri \kappa C) + \xi(k_L^{-1} V)
\right) k_L^{-1}.
\label{L2_constraint}
\ee
Upon introducing the shorthand notations
\be
\tilde{Y}_- = k_R^{-1} Y_- k_R \in \mfp
\quad \mbox{and} \quad 
\tilde{V} = k_L^{-1} V \in \bC^N, 
\ee
from the equation (\ref{L2_constraint}) 
it readily follows that
\be
\sinh(\ad_Q) \tilde{Y}_- 
= - \xi(\tilde{V}) 
+ \cosh(\ad_Q)(\ri \kappa C)
= - \xi(\tilde{V}) 
+ \ri \kappa \cosh(2 Q) C.
\label{tilde_Y}
\ee
Spelling out the components of the above matrix
equation, for all $k, l \in \bN_N$ we have
\be
\sinh(q_k - q_l) (\tilde{Y}_-)_{k, l}
= -\ri \mu (\tilde{V}_k \overline{\tilde{V}}_l 
- \delta_{k, l})
- \ri (\mu - \nu) C_{k, l} 
+ \ri \kappa \cosh(2 q_k) C_{k, l},
\label{Y_components}
\ee
where it is understood that $q_{n + c} = -q_c$
for all $c \in \bN_n$.

Now take an arbitrary $c \in \bN_n$. With the 
specialization $k = l = c$ the above equation 
(\ref{Y_components}) takes the form
\be
0 = - \ri \mu ( \vert \tilde{V}_c \vert^2 - 1),
\ee
whence it is obvious that 
$\tilde{V}_c = - \tilde{V}_{n + c} 
= e^{\ri \chi_c}$
with some real parameter $\chi_c \in \bR$. Notice 
also that with $k = c$ and $l = n + c$ the equation 
(\ref{Y_components}) translates into
\be
\sinh(2 q_c) (\tilde{Y}_-)_{c, n + c}
= \ri \nu + \ri \kappa \cosh(2 q_c).
\ee 
Therefore, under the assumption 
$\nu + \kappa \neq 0$, 
from the above relationship it is clear that 
$q_c \neq 0$.

Next, let $a, b \in \bN_n$ be arbitrary numbers 
satisfying $a \neq b$. With the specialization
$k = a$ and $l = b$ the equation 
(\ref{Y_components}) has the form
\be
\sinh(q_a - q_b) (\tilde{Y}_-)_{a, b}
= - \ri \mu \tilde{V}_a \overline{\tilde{V}}_b
\neq 0,
\ee
therefore $q_a \neq q_b$. Putting the above 
considerations together, we see that 
$q_1 > \ldots > q_n > 0$, whence the
\emph{regularity} property $q \in \mfc$ is 
immediate.

Now let us introduce the diagonal matrix
\be
m = \diag(e^{\ri \chi_1}, \ldots, e^{\ri \chi_n},
e^{\ri \chi_1}, \ldots, e^{\ri \chi_n}) \in M.
\label{L2_m}
\ee 
Due to the construction of $m$, it is clear 
that $m^{-1} \tilde{V} = E$ (\ref{E}).
Therefore, by applying the linear operator
$\Ad_{m^{-1}}$ on the equation (\ref{tilde_Y}), 
we obtain
\be
\sinh(\ad_Q) \Ad_{m^{-1}} (\tilde{Y}_-)
= - \xi(E) + \cosh(\ad_Q) (\ri \kappa C).
\ee
By inspecting the diagonal and the off-diagonal
parts of $\Ad_{m^{-1}} (\tilde{Y}_-)$ separately, 
and utilizing the regularity of $q$, we can write
\be
\Ad_{m^{-1}} (\tilde{Y}_-)
= P - \sinh(\tilde{\ad}_Q)^{-1} \xi(E)
+ \coth(\tilde{\ad}_Q)(\ri \kappa C)
\ee
with some diagonal matrix
$P = \diag(p_1, \ldots, p_n, -p_1, \ldots, -p_n) 
\in \mfa$. Thus, remembering the definition
(\ref{L_mfp}), we can simply write 
$\tilde{Y}_- = m L_\mfp(q, p) m^{-1}$; 
therefore the relationship 
\be
Y_- = k_R m L_\mfp(q, p) m^{-1} k_R^{-1}
\ee
is immediate. Since 
$V = k_L \tilde{V} = k_L m E$, we also have 
\be
\rho = \xi(V) = k_L m \xi(E) m^{-1} k_L^{-1}.
\ee
Therefore, with the group elements
$\eta_L = k_L m \in K$ and 
$\eta_R = k_R m \in K$, the lemma follows.
\end{proof}

To proceed further, let us notice that the 
Abelian group
\be
U(1)_* 
= \{ (e^{\ri \chi} \bsone_N, 
e^{\ri \chi} \bsone_N)
\in K \times K
\, | \, 
\chi \in \bR \}
\cong U(1)
\label{U(1)_*}
\ee
is a closed normal subgroup of the product Lie 
group $K \times K$, whence the coset space 
$(K \times K) / U(1)_*$ inherits a natural
(real) Lie group structure from $K \times K$.
Let us also consider the smooth product manifold
\be
\cM^S = \cP^S \times (K \times K) / U(1)_*.
\label{cM_S}
\ee
Having equipped with the above objects, now
we can introduce a natural parametrization of 
the level set $\mfL_0$ (\ref{mfL_0}) motivated 
by Lemma \ref{L2}. Indeed, by imitating the proof 
of Lemma 2 in \cite{Pusztai_NPB2011}, one can 
easily verify that the map
\be
\Upsilon^S \colon \cM^S \rightarrow \cP^\ext,
\quad
(q, p, (\eta_L, \eta_R) U(1)_*)
\mapsto
(\eta_L e^Q \eta_R^{-1},
\eta_R L(q, p) \eta_R^{-1},
\eta_L \xi(E) \eta_L^{-1})
\label{Upsilon_S}
\ee
is a well-defined injective immersion with 
image $\Upsilon^S(\cM^S) = \mfL_0$. 
Also, just as in the proof of Lemma 3 in 
\cite{Pusztai_NPB2011}, the parametrization
$\Upsilon^S$ makes it easy to verify directly
that the zero element of the Lie algebra
$\mfs(\mfk \oplus \mfk)$ (\ref{mfs_mfk_mfk}) 
is a regular value of the momentum map $J^\ext$
(\ref{J_ext}). Therefore the level set 
$\mfL_0$ is an \emph{embedded submanifold} of 
$\cP^\ext$ in a natural manner. More precisely, 
there is a unique smooth manifold structure on 
$\mfL_0$ such that the pair $(\mfL_0, \iota_0)$
with the tautological injection
\be
\iota_0 \colon \mfL_0 \hookrightarrow \cP^\ext
\ee
is an embedded submanifold of $\cP^\ext$. 
At this point let us notice that, due to
the relationship 
$\Upsilon^S(\cM^S)  = \iota_0 (\mfL_0)$,
the map $\Upsilon^S$ factors through 
$(\mfL_0, \iota_0)$; therefore there is a 
well-defined map
\be 
\Upsilon^S_0 \colon \cM^S \rightarrow \mfL_0,
\ee
such that 
$\Upsilon^S = \iota_0 \circ \Upsilon^S_0$.
Let us observe that, since $\iota_0$ is an 
embedding, the map $\Upsilon^S_0$ is automatically 
\emph{smooth} 
(see e.g. Theorem 1.32 in \cite{Warner}).
Now, since $\Upsilon^S_0$ is a smooth bijective
immersion from $\cM^S$ onto $\mfL_0$, and
since it acts between manifolds
of the same dimension, it is immediate that
$\Upsilon^S_0$ is a diffeomorphism. 
The above ideas can be summarized by saying
that the diagram
\be
\begin{split}
\xymatrix{
 & \cP^\ext & 
\\
\cM^S \ar[rr]_{\Upsilon^S_0}^{\cong}
\ar@{^{(}->}[ur]^{\Upsilon^S} 
& & \mfL_0 \ar@{_{(}->}[ul]_{\iota_0} 
}
\label{Upsilon_S_diagram}
\end{split}
\ee
is commutative.
In other words, the pair 
$(\cM^S, \Upsilon^S)$ provides an equivalent 
model for the smooth embedded submanifold 
$(\mfL_0, \iota_0)$.
 
Utilizing the model $(\cM^S, \Upsilon^S)$ of the
level set $\mfL_0$, in the following we complete
the symplectic reduction of 
$(\cP^\ext, \omega^\ext)$ 
at the zero value of the momentum map $J^\ext$.
For this purpose let us note that on the model 
space $\cM^S$ (\ref{cM_S}) the residual 
$K \times K$-action takes the form
\be
(k_L, k_R) \acts (q, p, (\eta_L, \eta_R) U(1)_*)
= (q, p, (k_L \eta_L, k_R \eta_R) U(1)_*).
\ee
Therefore it is obvious that the orbit space 
$\cM^S / (K \times K)$ can be naturally identified 
with the base manifold of the trivial principal 
$(K \times K) / U(1)_*$-bundle
\be
\pi^S \colon \cM^S \twoheadrightarrow \cP^S,
\quad
(q, p, (\eta_L, \eta_R) U(1)_*)
\mapsto (q, p).
\label{pi_S}
\ee
An immediate consequence of the above observation
is that the reduced symplectic manifold can be
identified as
\be
\cP^\ext /\! /_0 (K \times K)
\cong \cM^S / (K \times K) 
\cong \cP^S.
\ee
As is known from the theory of symplectic 
reductions, the reduced symplectic form 
$\omega^S \in \Omega^2(\cP^S)$ is uniquely
determined by the condition
\be
(\pi^S)^* \omega^S = (\Upsilon^S)^* \omega^\ext.
\label{omega_S_def}
\ee
However, since the derivative of $\Upsilon^S$ 
(\ref{Upsilon_S}) can be worked out explicitly, 
the computation of the pull-backs in 
(\ref{omega_S_def}) is almost trivial. 
Either doing the calculations by hand, or
remembering the results presented in 
\cite{FeherPusztai2007}, the following
theorem is immediate.

\begin{THEOREM}
\label{T3}
The reduced symplectic form can be written as
$\omega^S 
= 2 \sum_{c = 1}^n \ddd q_c \wedge \ddd p_c$.
That is to say, up to some trivial rescaling, 
the globally defined coordinate functions $q_c$, 
$p_c$ $(c \in \bN_n)$ form a Darboux system on 
the reduced manifold $\cP^S$.
\end{THEOREM}

\subsection{Solution algorithm for the Sutherland
model}
The goal of this subsection is to present a solution 
algorithm for a class of Hamiltonian systems in 
association with the family of the $\Ad$-invariant 
smooth functions defined on the Lie algebra $\mfg$. 
For, take an arbitrary $\Ad$-invariant smooth 
function $F \colon \mfg \rightarrow \bR$, i.e. 
we require
\be
F(y Y y^{-1}) = F(Y)
\qquad 
(\forall Y \in \mfg, \, \forall y \in G).
\label{F_invariant}
\ee
Now, let
\be
\pr_\mfg \colon \cP^\ext 
= G \times \mfg \times \cO
\twoheadrightarrow \mfg
\ee
denote the canonical projection onto  
$\mfg$; then it is obvious that the composite 
function $\pr_\mfg^* F = F \circ \pr_\mfg$ is 
a smooth $K \times K$-invariant function on 
$\cP^\ext$; therefore it survives the reduction. 
More precisely, the corresponding reduced 
Hamiltonian $(\pr_\mfg^* F)^S \in C^\infty(\cP^S)$ 
has the form
\be
(\pr_\mfg^* F)^S = F \circ L, 
\label{S_reduced_Hamiltonian}
\ee
as can be readily seen from the defining 
relationship
\be
(\pi^S)^* (\pr_\mfg^* F)^S 
= (\Upsilon^S)^* \pr_\mfg^* F.
\ee
It is a standard fact in reduction theory
that the Hamiltonian flows of the `unreduced'
Hamiltonian system 
$(\cP^\ext, \omega^\ext, \pr_\mfg^* F)$ staying
on the level space $\mfL_0$ project
onto the flows of the reduced system
$(\cP^S, \omega^S, F \circ L)$. However, finding 
the `unreduced' flows is a relatively simple
exercise; therefore this projection method gives 
rise to a natural and efficient solution algorithm 
for the reduced Hamiltonian system. 

To make the above observation precise, we need
the integral curves of the Hamiltonian vector
field $\bsX_{\pr_\mfg^* F} \in \mfX(\cP^\ext)$
generated by $\pr_\mfg^* F$. Recalling
$\omega^\ext$ (\ref{omega_ext}), from the 
defining relationship 
\be
\bsX_{\pr_\mfg^* F} \inner \omega^\ext
= \ddd (\pr_\mfg^* F)
\ee 
we find easily that 
at each point $(y, Y, \rho) \in \cP^\ext$ 
the Hamiltonian vector field has the form
\be
(\bsX_{\pr_\mfg^* F})_{(y, Y, \rho)}
= (y \nabla F(Y))_y \oplus 0_Y \oplus 0_\rho
\in T_{(y, Y, \rho)} \cP^\ext,
\label{Ham_vect_field_F}
\ee
where the gradient $\nabla F(Y) \in \mfg$ is
defined by the condition 
\be
\langle \nabla F(Y), \delta Y \rangle
= (\ddd F)_Y (\delta Y)
\qquad
(\forall \delta Y \in T_Y \mfg \cong \mfg).
\ee
From (\ref{Ham_vect_field_F}) it is clear that 
the induced Hamiltonian flows are \emph{complete}, 
having the form
\be
\bR \ni t 
\mapsto 
(y_0 e^{t \nabla F(Y_0)}, Y_0, \rho_0) 
\in \cP^\ext
\ee
with some $(y_0, Y_0, \rho_0) \in \cP^\ext$.
Therefore the reduced Hamiltonian 
flows are also complete. 

Now take an arbitrary flow 
\be
\bR \ni t 
\mapsto
(q(t), p(t)) \in \cP^S
\label{reduced_flow_S}
\ee
induced by the reduced Hamiltonian $F \circ L$,
and let $L_0 = L(q(0), p(0))$. Since the
unreduced flow passing through the point
$(e^{Q(0)}, L_0, \xi(E)) \in \mfL_0$
at $t = 0$ projects onto the reduced flow 
(\ref{reduced_flow_S}), from the definition of
the parametrization $\Upsilon^S$ (\ref{Upsilon_S}) 
it is clear that for each $t \in \bR$ there is 
some $(\eta_L(t), \eta_R(t)) \in K \times K$
such that
\be
(e^{Q(0)} e^{t \nabla F(L_0)}, 
L_0, \xi(E)) 
= (\eta_L(t) e^{Q(t)} \eta_R(t)^{-1},
\eta_R(t) L(q(t), p(t)) \eta_R(t)^{-1},
\eta_L(t) \xi(E) \eta_L(t)^{-1}).
\label{S_flow}
\ee
By comparing the $\mfg$-components of the above
equation we see that
\be
L_0 = \eta_R(t) L(q(t), p(t)) \eta_R(t)^{-1},
\ee
from where we conclude that during the time evolution 
of the reduced dynamics the Lax matrix $L$ (\ref{L}) 
undergoes an isospectral deformation. What is even 
more important, the $G$-component of equation
(\ref{S_flow}) leads to the relationship
\be
e^{Q(0)} e^{t \nabla F(L_0)} 
e^{t \nabla F(L_0)^*} e^{Q(0)}
= \eta_L(t) e^{2 Q(t)} \eta_L(t)^{-1},
\ee
which entails the spectral identification
\be
\sigma(e^{2 Q(t)})
= \sigma(e^{Q(0)} e^{t \nabla F(L_0)} 
e^{t \nabla F(L_0)^*} e^{Q(0)})
= \sigma(e^{2 Q(0)} e^{t \nabla F(L_0)} 
e^{t \nabla F(L_0)^*}).
\ee
Somewhat more informally one can say that the
matrix flow $Q(t)$, and so the trajectory $q(t)$,
can be recovered simply by diagonalizing the matrix 
flow
\be
t \mapsto 
e^{2 Q(0)} 
e^{t \nabla F(L_0)} 
e^{t \nabla F(L_0)^*}.
\label{S_matrix_flow}
\ee
Besides providing a nice solution algorithm,
the above observation can also be seen
as the starting point of the scattering 
theoretic analysis of the reduced Hamiltonian 
system $(\cP^S, \omega^S, F \circ L)$.
Indeed, by analyzing the temporal asymptotics
of the matrix flow (\ref{S_matrix_flow}), one
can understand the temporal asymptotics of the
trajectory $q(t)$ as well.

In order to establish the connection with the 
Sutherland many-particle systems, notice that 
the Hamiltonian of the $BC_n$ Sutherland model 
(\ref{H_S}) with three independent coupling 
constants can be realized as the reduced 
Hamiltonian induced by the $\Ad$-invariant 
quadratic function 
\be
F_2 (Y) = \frac{1}{4} \langle Y, Y \rangle 
= \frac{1}{4} \tr(Y^2)
\qquad
(Y \in \mfg).
\label{F_2}
\ee
Indeed, a simple computation reveals that 
\be
(\pr_\mfg^* F_2)^S = H^S,
\ee
with coupling constants
\be
g^2 = \mu^2, 
\quad
g_1^2 = \half \nu \kappa,
\quad
g_2^2 = \half (\nu - \kappa)^2.
\label{S_coupling_parameters_from_RSvD}
\ee
Note that if $\nu \kappa \geq 0$, then the
interaction is purely repulsive.
Conversely, if one starts with an arbitrary triple
of the non-negative Sutherland coupling constants 
$(g^2, g_1^2, g_2^2)$ satisfying the inequalities 
$g^2 > 0$ and $g_1^2 + g_2^2 > 0$, then the 
corresponding repulsive Sutherland model can be 
recovered e.g. by choosing the parameter triple 
$(\mu, \nu, \kappa)$ with components
\be
\mu = - \vert g \vert,
\quad
\nu 
= \frac{\vert g_2 \vert + \sqrt{g_2^2 + 4 g_1^2}}
{\sqrt{2}},
\quad
\kappa = \frac{2 \sqrt{2} g_1^2}
{\vert g_2 \vert + \sqrt{g_2^2 + 4 g_1^2}}.
\label{RSvD_coupling_parameters_from_S}
\ee
Notice that the above defined parameters satisfy 
the inequalities $\mu < 0$, $\nu > 0$ and
$\kappa \geq 0$, therefore the assumptions made 
in Lemmas \ref{L1} and \ref{L2} are automatically 
met. As for the solution algorithm of the Sutherland 
model, note that $\nabla F_2(L_0) = L_0 / 2$, 
whence from equation (\ref{S_matrix_flow}) we see 
at once that the trajectories can be determined by 
diagonalizing the matrix flow
\be
t \mapsto 
e^{2 Q(0)} e^{\half t L_0} e^{\half t L_0^*}.
\label{S_matrix_flow_spec}
\ee

We close this section with some remarks on the
range of the attainable Sutherland coupling 
constants coming from the proposed symplectic 
reduction picture. From equation 
(\ref{RSvD_coupling_parameters_from_S}) it is clear 
that \emph{any} non-negative triple of the 
coupling constants $(g^2, g_1^2, g_2^2)$ with 
$g^2 > 0$ and $g_1^2 + g_2^2 > 0$ can be realized 
by an appropriate choice of the parameters 
$(\mu, \nu, \kappa)$ satisfying $\nu \kappa \geq 0$. 
Due to the repulsive nature of the interaction, 
in these cases the Sutherland model has only 
scattering states. It conforms with Lemma \ref{L1},
according which the eigenvalues of the Lax 
operator $L$ are real. Thus the scattering 
properties of the model can be understood 
explicitly by analyzing the temporal asymptotics 
of the truly exponential-type matrix flow 
(\ref{S_matrix_flow_spec}). 

However, looking at 
(\ref{S_coupling_parameters_from_RSvD}), 
it is obvious that from the proposed reduction 
picture we can derive the Sutherland model even 
with $g_1^2 < 0$, provided that $\nu \kappa < 0$. 
Since Lemma \ref{L1} does not apply in these cases, 
in certain non-empty region of the phase space 
$\cP^S$ some of the eigenvalues of the Lax operator 
may be purely imaginary. Therefore, as can be 
conjectured from (\ref{S_matrix_flow_spec}), 
some of the Sutherland particles may exhibit  
oscillatory behavior with bounded trajectories. 
To fully understand the details of this phenomenon, 
one should sharpen Lemma \ref{L1} to characterize 
the spectral properties of the Lax matrix for 
$\nu \kappa < 0$. Nevertheless, since our primary 
interest is to study the duality properties of the
standard hyperbolic Sutherland $BC_n$ model with 
purely repulsive interaction, we leave this 
interesting exercise for a future study.

\section{The rational $BC_n$ RSvD model}
\label{S4}
\setcounter{equation}{0}
In this section we work out the rational $BC_n$ RSvD 
model with three independent coupling constants 
from the proposed symplectic reduction framework. 
Let us keep in mind that Section \ref{S3} and the 
present section provide the necessary technical 
background to establish the action-angle duality 
between the standard repulsive Sutherland model 
(\ref{H_S}) and the rational RSvD (\ref{H_R}) 
model with $\nu \kappa \geq 0$. Therefore, 
according to our discussion in the concluding 
remarks of Section \ref{S3}, we may assume at the 
outset that $\mu < 0$, $\nu > 0$ and 
$\kappa \geq 0$. 

\subsection{The phase space of the RSvD model}
\setcounter{equation}{0}
First of all, we need some objects introduced in 
the study \cite{Pusztai_NPB2011} of the $C_n$-type 
rational RSvD model associated with the pair of
non-zero parameters $(\mu, \nu)$. For each 
$a \in \bN_n$ let us consider the complex-valued 
rational function 
\be
\mfc \ni \lambda 
= (\lambda_1, \ldots, \lambda_n)
\mapsto
z_a(\lambda) 
= - \left(1 + \frac{\ri \nu}
{\lambda_a} \right)
\prod_{\substack{d = 1 \\ (d \neq a)}}^n 
\left( 1 + \frac{2 \ri \mu}
{\lambda_a - \lambda_d} \right)
\left( 1 + \frac{2 \ri \mu}{
\lambda_a + \lambda_d} \right) \in \bC.
\label{z}
\ee
Recall also that the Lax matrix of the rational 
$C_n$ RSvD model is defined by the matrix-valued 
function
\be
\cA \colon 
\cP^R \rightarrow \exp(\mfp),
\quad
(\lambda, \theta) \mapsto \cA(\lambda, \theta),
\label{cA_def}
\ee
where the matrix entries lying in the diagonal
$n \times n$ blocks are given by the formulae
\begin{align}
& \cA_{a, b}(\lambda, \theta) 
= e^{\theta_a + \theta_b}
\vert z_a(\lambda) z_b(\lambda) \vert^\half 
\frac{2 \ri \mu}{2 \ri \mu + \lambda_a - \lambda_b},
\\
& \cA_{n + a, n + b}(\lambda, \theta) 
= e^{-\theta_a - \theta_b}
\frac{ \overline{z_a(\lambda)} z_b(\lambda)}
{\vert z_a(\lambda) z_b(\lambda) \vert^\half} 
\frac{2 \ri \mu}{2 \ri \mu - \lambda_a + \lambda_b},
\end{align}
meanwhile the matrix entries belonging to the
off-diagonal $n \times n$ blocks have the form
\be 
\cA_{a, n + b}(\lambda, \theta) 
= \overline{\cA_{n + b, a}(\lambda, \theta)} 
= e^{\theta_a - \theta_b} z_b(\lambda) 
\vert z_a(\lambda) z_b(\lambda)^{-1} \vert^\half
\frac{2 \ri \mu}{2 \ri \mu + \lambda_a + \lambda_b}
+ \frac{\ri(\mu - \nu)}{\ri \mu + \lambda_a} 
\delta_{a, b},
\label{cA_entries}
\ee
for any $a, b \in \bN_n$. Besides the Lax matrix 
$\cA$, from the theory of the $C_n$ RSvD 
model we need the column vector 
$\cF(\lambda, \theta) \in \bC^N$ with components
\be 
\cF_a(\lambda, \theta) 
= e^{\theta_a} \vert z_a(\lambda) \vert^\half
\quad \mbox{and} \quad
\cF_{n + a}(\lambda, \theta) 
= e^{- \theta_a} \overline{z_a(\lambda)} 
\vert z_a(\lambda) \vert^{-\half},
\label{cF}
\ee
where $a \in \bN_n$. Also, in the forthcoming
computations we shall frequently encounter 
the column vector
\be
\cV(\lambda, \theta) 
= \cA(\lambda, \theta)^{-\half} 
\cF(\lambda, \theta)
\in \bC^N.
\label{cV}
\ee
Notice that $\cV(\lambda, \theta)$ is 
well-defined, since the positive definite Lax 
matrix $\cA(\lambda, \theta) \in \exp(\mfp)$ has
a unique square root belonging to $\exp(\mfp)$.
Furthermore, as it can be seen from Proposition 8 
in \cite{Pusztai_NPB2011}, we have $\cV^* \cV = N$
and $C \cV + \cV = 0$.

In order to handle the $BC_n$ RSvD 
model associated with the parameter triple  
$(\mu, \nu, \kappa)$, we need some new objects, 
too. In particular, let us introduce the smooth
functions $\alpha$ and $\beta$ defined on the 
\emph{positive half-line} by the formulae
\be
\alpha(x) = \frac{\sqrt{x + \sqrt{x^2 + \kappa^2}}}
{\sqrt{2 x}}
\quad \mbox{and} \quad
\beta(x) = \ri \kappa \frac{1}{\sqrt{2 x}}
\frac{1}{\sqrt{x + \sqrt{x^2 + \kappa^2}}},
\label{alpha&beta}
\ee
where $x \in ( 0, \infty )$. Now, with each
positive $n$-tuple
$\lambda = (\lambda_1, \ldots, \lambda_n) 
\in (0, \infty)^n$
we associate the $n \times n$ diagonal matrix
$\bslambda = \diag(\lambda_1, \ldots, \lambda_n)$
and consider the Hermitian $N \times N$ matrix
\be
h(\lambda) = \begin{bmatrix}
\alpha(\bslambda) & \beta(\bslambda) \\
-\beta(\bslambda) & \alpha(\bslambda)
\end{bmatrix}.
\label{h}
\ee
Making use of the functional equation 
$\alpha(x)^2 + \beta(x)^2 = 1$, one can show 
that $h(\lambda) C h(\lambda) = C$, i.e. the matrix  
$h(\lambda)$ is an Hermitian element of the Lie 
group $G$ (\ref{G}). Finally, let us introduce the 
shorthand notation
\be
\Lambda = \diag(\lambda_1, \ldots, \lambda_n,
-\lambda_1, \ldots, -\lambda_n) \in \mfa. 
\ee
Having equipped with the above objects, we 
are now in a position to provide an appropriate 
parametrization of the level set $\mfL_0$ 
(\ref{mfL_0}) based on the diagonalization of 
the Lie algebra part of $\mfL_0$.

\begin{LEMMA}
\label{L4}
Suppose that $\nu \neq 2 \mu$, $\nu + \kappa \neq 0$ 
and $\nu \kappa \geq 0$; then for each point 
$(y, Y, \rho) \in \mfL_0$ there are some 
$\lambda \in \mfc$, $\theta \in \bR^n$ and 
$\eta_L, \eta_R \in K$, such that
\be
y = \eta_L \cA(\lambda, \theta)^\half 
h(\lambda)^{-1} \eta_R^{-1},
\quad
Y = \eta_R h(\lambda) 
\Lambda h(\lambda)^{-1} \eta_R^{-1},
\quad
\rho = \eta_L \xi(\cV(\lambda, \theta)) 
\eta_L^{-1}.
\ee
\end{LEMMA}

\begin{proof}
Take an arbitrary point $(y, Y, \rho) \in \mfL_0$.
Remembering the momentum map $J^\ext$ (\ref{J_ext}), 
it is clear that $Y_+ = - \ri \kappa C$. 
On the other hand, since any element of the subspace
$\mfp$ (\ref{mfkp}) can be conjugated into $\mfa$ 
(\ref{mfa}) by some element of $K$ (\ref{K}), we 
can write 
\be
Y_- = k_R D k_R^{-1} 
\ee
with some $k_R \in K$ and
\be 
D = \diag(d_1, \ldots, d_n, -d_1, \ldots, -d_n) 
\in \mfa 
\ee
satisfying $d_1 \geq \ldots \geq d_n \geq 0$.
Therefore 
\be
Y = k_R (D - \ri \kappa C) k_R^{-1},
\label{Y_p1}
\ee
from where we obtain 
\be
Y^2 
= k_R (D^2 - \kappa^2 \bsone_N) k_R^{-1}.
\label{L4_Y2}
\ee
However, by combining Lemmas \ref{L1} and \ref{L2}, 
it is obvious that $Y^2$ is \emph{positive definite}, 
i.e. $Y^2 > 0$. Giving a glance at the
above equation (\ref{L4_Y2}), it is thus clear 
that $\lambda_a = \sqrt{d_a^2 - \kappa^2}$ is 
a well-defined \emph{positive} real number for 
each $a \in \bN_n$, satisfying the inequalities
$\lambda_1 \geq \ldots \geq \lambda_n > 0$.
Now, recalling $h(\lambda)$
(\ref{h}) and utilizing the functional equations
\be
\alpha(x)^2 - \beta(x)^2 
= \sqrt{1 + \frac{\kappa^2}{x^2}}
\quad \mbox{and} \quad
2 \alpha(x) \beta(x) = \frac{\ri \kappa}{x},
\ee
it is not hard to see that
\be 
h(\lambda) \Lambda h(\lambda)^{-1}
= D - \ri \kappa C.
\ee
Plugging this observation into (\ref{Y_p1}), 
we get
\be
Y = k_R h(\lambda) \Lambda 
h(\lambda)^{-1} k_R^{-1}.
\label{Y_p_OK}
\ee
To proceed further, notice that 
the momentum map constraint yields
the relationship
\be
0 = (y Y y^{-1})_+ + \rho 
= \frac{y Y y^{-1} - (y Y y^{-1})^*}{2} + \rho
\ee  
as well. This is clearly equivalent
to the equation
\be
y^* y Y - Y^* y^* y + 2 y^* \rho y = 0.
\label{constraint_1}
\ee 
However, due to (\ref{cO}) we can write 
$\rho = \xi(V)$, where $V \in \bC^N$ is an 
appropriate column vector satisfying $V^* V = N$ 
and $C V + V = 0$. Thus, recalling (\ref{xi}), 
the above constraint (\ref{constraint_1}) 
can be cast into the form
\be
y^* y Y - Y^* y^* y + 2 \ri \mu y^* V V^* y
- 2 \ri \mu y^* y + 2 \ri (\mu - \nu) C = 0.
\ee
Plugging the parametrization (\ref{Y_p_OK}) into 
the above equation, we get
\be
2 \ri \mu h k_R^* y^* y k_R h
+ \Lambda h k_R^* y^* y k_R h
- h k_R^* y^* y k_R h \Lambda
= 2 \ri \mu (h k_R^* y^* V) (h k_R^* y^* V)^*
+ 2 \ri (\mu - \nu) C.
\label{key}
\ee
At this point let us note that the last equation 
can be naturally identified with equation (31) in 
\cite{Pusztai_JPA2011}. Since this equation
was the starting point of our scattering theoretic
analysis of the hyperbolic $C_n$ Sutherland model,
we have full control over the structure of
its ingredients. In particular, due to Lemma 1
in \cite{Pusztai_JPA2011}, we know that $\Lambda$
must be a \emph{regular} element of $\mfa$, whence
$\lambda_1 > \ldots > \lambda_n > 0$, i.e.
$\lambda \in \mfc$. Furthermore, by Lemma 2
in \cite{Pusztai_JPA2011}, we can write
\be
h(\lambda) k_R^* y^* y k_R h(\lambda) 
= m \cA(\lambda, \theta) m^*
\ee
with some $\theta \in \bR^n$ and $m \in M$. 
Noticing that the group element $h(\lambda) \in G$ 
commutes with each element of $M$ (\ref{M}), we 
obtain 
\be
(y k_R m h)^* y k_R m h
= m^* h k_R^* y^* y k_R h m
= \cA.
\label{A_p}
\ee
Third, our analysis yields the relationship 
\be
h(\lambda) k_R^* y^* V = m \cF(\lambda, \theta)
\label{V_p}
\ee
as well. 
Now let us notice that equation (\ref{A_p}) 
allows us to give a characterization of the 
global Cartan decomposition (polar decomposition) 
of $y k_R m h(\lambda)$. Indeed, we have 
\be
y k_R m h(\lambda) 
= \eta_L \cA(\lambda, \theta)^\half
\ee 
with some $\eta_L \in K$; thus the parametrization
\be
y 
= \eta_L \cA(\lambda, \theta)^\half 
h(\lambda)^{-1} (k_R m)^{-1}
\label{y_param}
\ee
is immediate. Plugging this into (\ref{V_p}) and
remembering the definition (\ref{cV}), we obtain 
at once that $V = \eta_L \cV(\lambda, \theta)$; 
therefore the relationship
\be
\rho = \xi(V) 
= \eta_L \xi(\cV(\lambda, \theta)) \eta_L^{-1}
\label{rho_param}
\ee
also follows. Finally, upon setting 
$\eta_R = k_R m \in K$, from the equations 
(\ref{y_param}), (\ref{Y_p_OK}) and 
(\ref{rho_param}) we see that the 
proof is complete.
\end{proof}

Motivated by the above lemma, we can introduce
a natural parametrization of the level set
$\mfL_0$ (\ref{mfL_0}). To make this idea 
precise, first let us define the smooth product 
manifold
\be
\cM^R = \cP^R \times (K \times K) / U(1)_*.
\label{cM_R}
\ee
Now, a trivial generalization of the proof of 
Lemma 10 in \cite{Pusztai_NPB2011} immediately 
convinces us that the smooth map 
\be
\Upsilon^R \colon \cM^R \rightarrow \cP^\ext
\ee
defined by the assignment
\be
(\lambda, \theta, (\eta_L, \eta_R) U(1)_*)
\mapsto
(\eta_L \cA(\lambda, \theta)^\half
h(\lambda)^{-1} \eta_R^{-1},
\eta_R h(\lambda) \Lambda 
h(\lambda)^{-1} \eta_R^{-1},
\eta_L \xi(\cV(\lambda, \theta)) \eta_L^{-1})
\label{Upsilon_R}
\ee
is a well-defined injective immersion with image 
$\Upsilon^R(\cM^R) = \mfL_0$. Repeating the same
arguments that we applied in the Sutherland picture, 
it is clear that $\Upsilon^R$ factors through the
embedded submanifold $(\mfL_0, \iota_0)$, and the 
resulting smooth map 
$\Upsilon^R_0 \colon \cM^R \rightarrow \mfL_0$
is a diffeomorphism. That is to say, the pair
$(\cM^R, \Upsilon^R)$ provides an equivalent
model for the smooth embedded submanifold 
$(\mfL_0, \iota_0)$.

In order to complete the reduction of $\cP^\ext$ 
at the zero value of the momentum map $J^\ext$, 
let us observe that the residual 
$K \times K$-action on the model space $\cM^R$ 
(\ref{cM_R}) takes the form
\be
(k_L, k_R) 
\acts 
(\lambda, \theta, (\eta_L, \eta_R)U(1)_*)
= (\lambda, \theta, 
(k_L \eta_L, k_R \eta_R) U(1)_*).
\ee
Therefore the orbit space $\cM^R / (K \times K)$
gets naturally identified with the base manifold
of the trivial principal 
$(K \times K) / U(1)_*$-bundle
\be
\pi^R \colon \cM^R \twoheadrightarrow \cP^R,
\quad
(\lambda, \theta, (\eta_L, \eta_R)U(1)_*)
\mapsto
(\lambda, \theta).
\label{pi_R}
\ee
It is thus evident that for the reduced
symplectic manifold we have the alternative
identification
\be
\cP^\ext /\!/_0 (K \times K)
\cong
\cM^R / (K \times K)
\cong
\cP^R.
\ee

Our remaining task is to make explicit the reduced 
symplectic form $\omega^R \in \Omega^2(\cP^R)$
naturally induced on the reduced phase space $\cP^R$. 
Just as in the Sutherland case, it is a tempting 
idea to compute the reduced symplectic structure 
via a formula analogous to (\ref{omega_S_def}). 
However, due to the presence of the square root of 
$\cA$ in $\Upsilon^R$ (\ref{Upsilon_R}), this 
approach seems to be hopeless. Instead, it is
more expedient to invoke the alternative machinery 
presented in Subsection 4.3 of paper 
\cite{Pusztai_NPB2011}. 
Namely, by analyzing the Poisson brackets
of the auxiliary $K \times K$-invariant functions 
defined in equations (4.65) and (4.66) of 
\cite{Pusztai_NPB2011}, an almost verbatim 
computation as in the $C_n$ case convinces us that 
the standard coordinates of $\cP^R$ are canonical.

\begin{THEOREM}
\label{T5}
Utilizing the global coordinate functions
$\lambda_c$ and $\theta_c$ $(c \in \bN_n)$
defined on the reduced manifold $\cP^R$, the 
reduced symplectic structure takes the form
$\omega^R 
= 2 \sum_{c = 1}^n 
\ddd \theta_c \wedge \ddd \lambda_c$.
\end{THEOREM} 

\subsection{Solution algorithm for the RSvD model}
In this subsection we work out an efficient 
solution algorithm for the rational $BC_n$ RSvD 
model built on the projection method naturally 
offered by the symplectic reduction framework. 
For, take an arbitrary $K \times K$-invariant 
real-valued smooth function 
$f \colon G \rightarrow \bR$ defined on the Lie
group $G$, i.e. we assume that
\be
f((k_L, k_R) \acts y)
= f(k_L y k_R^{-1}) 
= f(y)
\qquad
(\forall y \in G, \,
\forall (k_L, k_R) \in K \times K). 
\label{f_invariant}
\ee
Let
\be
\pr_G \colon \cP^\ext 
= G \times \mfg \times \cO
\twoheadrightarrow G
\label{pr_G}
\ee
denote the canonical projection onto $G$;
then the composite function
$\pr_G^* f = f \circ \pr_G$ is clearly
$K \times K$-invariant on the unreduced
phase space $\cP^\ext$. Thus, based on 
the defining formula 
\be
(\pi^R)^*(\pr_G^* f)^R = (\Upsilon^R)^* \pr_G^* f,
\ee
it is easy to see that the corresponding reduced 
Hamiltonian $(\pr_G^* f)^R \in C^\infty(\cP^R)$ 
has the form
\be
(\pr_G^* f)^R (\lambda, \theta)
= f( \cA(\lambda, \theta)^\half h(\lambda)^{-1})
\qquad
((\lambda, \theta) \in \cP^R).
\label{R_reduced_Hamiltonian}
\ee
As we have discussed it in the Sutherland
picture, the essence of the projection method is
that the flows of the unreduced Hamiltonian system 
staying on the level space $\mfL_0$ project onto
the flows of the reduced Hamiltonian system.
However, from the defining formula
\be
\bsX_{\pr_G^* f} \inner \omega^\ext 
= \ddd (\pr_G^* f),
\ee 
we see immediately that at
each $(y, Y, \rho) \in \cP^\ext$ the
Hamiltonian vector field $\bsX_{\pr_G^* f}$
generated by $\pr_G^* f$ has the form
\be
(\bsX_{\pr_G^* f})_{(y, Y, \rho)}
= 0_y \oplus (- \nabla f(y))_Y \oplus 0_\rho
\in T_{(y, Y, \rho)} \cP^\ext,
\ee
where the gradient $\nabla f(y) \in \mfg$ 
is defined by the requirement
\be
\langle \nabla f(y), y^{-1} \delta y \rangle
= (\ddd f)_y (\delta y)
\qquad
(\forall \delta y \in T_y G).
\label{G_gradient}
\ee
Therefore the Hamiltonian flows of $\pr_G^* f$
are \emph{complete}, having the very simple form
\be
\bR \ni t 
\mapsto 
(y_0, Y_0 - t \nabla f(y_0), \rho_0) \in \cP^\ext,
\label{R_unreduced_flow}
\ee
where $(y_0, Y_0, \rho_0) \in \cP^\ext$ is an
arbitrary point. Hence the reduced flows are 
complete as well. 

Now take an arbitrary flow
\be
\bR \ni t 
\mapsto 
(\lambda(t), \theta(t)) \in \cP^R
\label{R_reduced_flow}
\ee
of the reduced Hamiltonian system
$(\cP^R, \omega^R, (\pr_G^* f)^R)$.
For simplicity let us now introduce the 
shorthand notations
\be
\cA_0 = \cA(\lambda(0), \theta(0)),
\quad
h_0 = h(\lambda(0)),
\quad
\Lambda_0 = \Lambda(0),
\quad
\cV_0 = \cV(\lambda(0), \theta(0)).
\ee
It is obvious that the unreduced flow
\be
\bR \ni t
\mapsto
(\cA_0^\half h_0^{-1}, 
h_0 \Lambda_0 h_0^{-1} 
- t \nabla f(\cA_0^\half h_0^{-1}),
\xi(\cV_0)) \in \mfL_0
\ee
projects onto (\ref{R_reduced_flow}). Recalling
$\Upsilon^R$ (\ref{Upsilon_R}), we see that for
each $t \in \bR$ we can find some pair of group
elements $(\eta_L(t), \eta_R(t)) \in K \times K$
such that
\begin{align}
& \cA_0^\half h_0^{-1} 
= \eta_L(t) \cA(\lambda(t), \theta(t))^\half 
h(\lambda(t))^{-1} \eta_R(t)^{-1}, 
\label{R_G_component}
\\
& h_0 \Lambda_0 h_0^{-1} 
- t \nabla f(\cA_0^\half h_0^{-1})
= \eta_R(t) h(\lambda(t)) \Lambda(t) 
h(\lambda(t))^{-1} \eta_R(t)^{-1}, 
\label{R_mfg_component}
\\
& \xi(\cV_0) 
= \eta_L(t) 
\xi(\cV(\lambda(t), \theta(t))) 
\eta_L(t)^{-1}.
\label{R_cO_component}
\end{align}
From (\ref{R_G_component}) we conclude that
\be
h_0^{-1} \cA_0 h_0^{-1}
= \eta_R(t) h(\lambda(t))^{-1} 
\cA(\lambda(t), \theta(t)) 
h(\lambda(t))^{-1} \eta_R(t)^{-1},
\ee
which entails the spectral identification
\be
\sigma(h_0^{-1} \cA_0 h_0^{-1})
= \sigma(h(\lambda(t))^{-1} 
\cA(\lambda(t), \theta(t)) 
h(\lambda(t))^{-1}). 
\ee
Thus, during the time evolution of the reduced
Hamiltonian system, the positive definite 
Hermitian matrix
\be
\cA^\bc(\lambda, \theta) 
= h(\lambda)^{-1} 
\cA(\lambda, \theta) 
h(\lambda)^{-1}
\in G
\label{R_Lax_matrix}
\ee
undergoes an \emph{isospectral} deformation.
Meanwhile, from (\ref{R_mfg_component}) it 
follows that
\be
\sigma(h_0 \Lambda_0 h_0^{-1} 
- t \nabla f(\cA_0^\half h_0^{-1}))
= \sigma(\Lambda(t))
= \{ \lambda_1(t), \ldots, \lambda_n(t),
- \lambda_1(t), \ldots, - \lambda_n(t) \},
\ee
whence the trajectory $t \mapsto \lambda(t)$ can
be recovered simply by diagonalizing the 
\emph{linear} matrix
flow
\be
t \mapsto
h_0 \Lambda_0 h_0^{-1} 
- t \nabla f(\cA_0^\half h_0^{-1}).
\label{R_matrix_flow}
\ee
It is worth mentioning that, due to its linearity
in $t$, the temporal asymptotics of the above matrix 
flow (\ref{R_matrix_flow}) can be analyzed by 
elementary perturbation theoretic techniques. Note 
that this observation could serve as the starting 
point of the scattering theoretic analysis of the 
reduced Hamiltonian system 
$(\cP^R, \omega^R, (\pr_G^* f)^R)$.

Our considerations so far apply to any reduced
system associated with a  
$K \times K$-invariant smooth function $f$ 
(\ref{f_invariant}). Note, however, that under
the assumption $\nu \kappa \geq 0$ the 
rational $BC_n$ RSvD model with three independent 
coupling constants can be nicely fitted into this 
picture. Indeed, upon introducing the 
$K \times K$-invariant function
\be
f_1(y) = \half \tr(y y^*)
\qquad
(y \in G),
\label{f_1}
\ee
one can verify that the corresponding reduced
Hamiltonian coincides with the RSvD Hamiltonian
(\ref{H_R}) associated with the coupling 
parameters $(\mu, \nu, \kappa)$, i.e.
\be
(\pr_G^* f_1)^R = H^R.
\label{f_1_reduced}
\ee
Let us observe that the assumption 
$\nu \kappa \geq 0$ automatically guarantees the 
lower bound $H^R > n$ on the RSvD Hamiltonian.
We mention in passing that the verification of 
(\ref{f_1_reduced}) is a quite tedious, but 
elementary calculation. Nevertheless, it can be 
done easily by utilizing the functional identities 
collected in the appendix of 
\cite{Pusztai_NPB2011}. 

Turning to the solution algorithm based on
(\ref{R_matrix_flow}), notice that for the 
gradient (\ref{G_gradient}) of the function 
$f_1$ we have
\be
\nabla f_1(y)
= \half \left(
y^*y - (y^* y)^{-1}
\right) \in \mfp
\qquad
(y \in G).
\ee
Therefore, from (\ref{R_Lax_matrix}) and 
(\ref{R_matrix_flow}) we see at once that the 
trajectories of the rational $BC_n$ RSvD model 
can be determined by diagonalizing the matrix 
flow
\be
t \mapsto h_0 \Lambda_0 h_0^{-1}
- \half t 
\left( \cA^\bc_0 - (\cA^\bc_0)^{-1} \right),
\label{RSvD_flows_spec}
\ee
where 
$\cA^\bc_0 = \cA^\bc(\lambda(0), \theta(0))$.
We find it remarkable that the properties of 
dynamics generated by the highly non-trivial 
Hamiltonian $H^R$ (\ref{H_R}) can be 
captured by analyzing the linear matrix flow
(\ref{RSvD_flows_spec}).

To sum up, we see that under the assumption
$\nu \kappa \geq 0$ the rational $BC_n$
RSvD model (\ref{H_R}) can be derived from 
an appropriate symplectic reduction 
framework. Parallel to our discussion 
in Section \ref{S3} on 
the Sutherland model (\ref{H_S}) with 
$g_1^2 < 0$, we expect that the RSvD model 
with $\nu \kappa < 0$ can also be understood 
from symplectic reduction by generalizing 
Lemma \ref{L1}. We wish to come back to this 
issue in a later publication.

\section{Discussion}
\label{S5}
\setcounter{equation}{0}
In the previous two sections we derived both the
standard hyperbolic $BC_n$ Sutherland and the 
rational $BC_n$ RSvD models from a unified 
symplectic reduction framework. The derivation of 
the Sutherland model relies on the $KAK$ 
decomposition of the Lie group part of the level 
set $\mfL_0$ (\ref{mfL_0}), meanwhile the 
symplectic geometric understanding of the RSvD 
model builds upon the diagonalization of the Lie 
algebra part of $\mfL_0$. Thereby, by performing 
the Marsden--Weinstein reduction of the symplectic 
manifold $(\cP^\ext, \omega^\ext)$ at the zero 
value of the momentum map $J^\ext$ (\ref{J_ext}), 
we end up with two equivalent realizations, $\cP^S$ 
and $\cP^R$, of the \emph{same} symplectic quotient 
$\cP^\ext /\!/_0 (K \times K)$. Thus it is obvious 
that there is a natural symplectomorphism
$\cS \colon \cP^S \rightarrow \cP^R$ 
making the diagram
\be
\begin{split}
\xymatrix{
 & \cP^\ext & 
\\
\cM^S \ar[r]^{\Upsilon^S_0}_{\cong} 
\ar@{>>}[d]_{\pi^S} 
& \mfL_0 \ar@{^{(}->}[u]^{\iota_0} 
& \ar[l]_{\Upsilon^R_0}^{\cong} 
\ar@{>>}[d]^{\pi^R} \cM^R 
\\
\cP^S \ar[rr]_{\cS}^{\cong} & & \cP^R
}
\label{diagram}
\end{split}
\ee
commutative. In the rest of this section we examine 
some of the immediate consequences of the dual 
reduction picture (\ref{diagram}). First, let us 
define the functions
\be
\hat{\lambda}_c = \cS^* \lambda_c,
\quad
\hat{\theta}_c = \cS^* \theta_c,
\quad
\check{q}_c = (\cS^{-1})^* q_c,
\quad
\check{p}_c = (\cS^{-1})^* p_c,
\label{new_coordinates}
\ee
where $c \in \bN_n$. Since $\cS$ is a 
symplectomorphism, from Theorems \ref{T3} and 
\ref{T5} it follows that
\be
\omega^S = \cS^* \omega^R
= 2 \sum_{c = 1}^n 
\ddd \hat{\theta}_c \wedge \ddd \hat{\lambda}_c 
\quad \mbox{and} \quad
\omega^R = (\cS^{-1})^* \omega^S
= 2 \sum_{c = 1}^n 
\ddd \check{q}_c \wedge \ddd \check{p}_c.
\ee
In other words, the globally defined functions
$\hat{\lambda}_c$ and $\hat{\theta}_c$ provide 
a new Darboux system on the Sutherland phase space 
$\cP^S$, meanwhile the global coordinates 
$\check{q}_c$ and $\check{p}_c$ are canonical
on the RSvD phase space $\cP^R$. Now
we show that these new families of canonical 
coordinates give rise to natural action-angle 
variables for the Sutherland and the RSvD models, 
respectively.

Starting with the Sutherland side of the dual
reduction picture, take an arbitrary point
$(q, p) \in \cP^S$ and let 
$(\lambda, \theta) = \cS(q, p) \in \cP^R$. Now,
recalling the parametrizations $\Upsilon^S$
(\ref{Upsilon_S}) and $\Upsilon^R$ 
(\ref{Upsilon_R}), from the commutativity of
the diagram (\ref{diagram}) it is clear that
\be
(e^Q, L(q, p), \xi(E))
= (\eta_L \cA(\lambda, \theta)^\half 
h(\lambda)^{-1} \eta_R^{-1},
\eta_R h(\lambda) \Lambda 
h(\lambda)^{-1} \eta_R^{-1},
\eta_L \xi(\cV(\lambda, \theta)) 
\eta_L^{-1})
\label{S_key}
\ee
with some group elements $\eta_L, \eta_R \in K$. 
By inspecting the $\mfg$-component of the above 
equation we obtain the spectral identification
\be
\sigma(L(q, p)) = \sigma(\Lambda)
= \{ \pm \hat{\lambda}_c(q, p) 
\, | \, 
c \in \bN_n \}.
\label{S_picture_spectral}
\ee

Next, take an arbitrary real-valued $\Ad$-invariant 
smooth function $F \colon \mfg \rightarrow \bR$ 
(\ref{F_invariant}) defined on 
$\mfg$, and consider the naturally 
associated reduced Hamiltonian system 
$(\cP^S, \omega^S, (\pr_\mfg^* F)^S)$. Due
to (\ref{S_reduced_Hamiltonian}) and (\ref{S_key}) 
it is clear that
\be
(\pr_\mfg^* F)^S = F \circ L
= F \circ \Lambda
= F (\diag(\hat{\lambda}_1, \ldots, \hat{\lambda}_n,
-\hat{\lambda}_1, \ldots, -\hat{\lambda}_n)),
\ee
i.e. the reduced Hamiltonian depends only on the
coordinates $\hat{\lambda}_c$ $(c \in \bN_n)$.
It follows that the global coordinates
$\hat{\lambda}_c$ and $\hat{\theta}_c$ 
$(c \in \bN_n)$ provide canonical 
\emph{action-angle variables} for the 
reduced system 
$(\cP^S, \omega^S, (\pr_\mfg^* F)^S)$.
Note that the action and the angle coordinates 
of the Sutherland picture are exactly the 
pull-backs of the canonical positions and the
canonical momenta of the Ruijsenaars picture. 

To conclude the study of the Sutherland side of
the dual reduction picture (\ref{diagram}), let 
us recall that the hyperbolic $BC_n$ Sutherland 
model can be realized as the reduced Hamiltonian 
system generated by the quadratic $\Ad$-invariant
function $F_2$ (\ref{F_2}). Therefore the above 
construction of action-angle coordinates applies 
to the repulsive Sutherland model equally well. 
Also, from (\ref{S_picture_spectral}) we see that 
the positive eigenvalues of $L$ (\ref{L}) provide 
$n$ functionally independent first integrals in 
involution. Furthermore, the matrix $L$ naturally 
enters the solution algorithm of the Sutherland 
model (see equation (\ref{S_matrix_flow_spec})), 
whence the non-Hermitian matrix $L$ is indeed a 
\emph{Lax matrix} for the Sutherland many-particle
system.

In the following we turn our attention to the 
Ruijsenaars side of the dual reduction picture
(\ref{diagram}). For, take an arbitrary point
$(\lambda, \theta) \in \cP^R$ and let
$(q, p) = \cS^{-1}(\lambda, \theta) \in \cP^S$.
Recalling the mappings $\Upsilon^R$ 
(\ref{Upsilon_R}) and $\Upsilon^S$ 
(\ref{Upsilon_S}), it is clear that
\be
(\cA(\lambda, \theta)^\half h(\lambda)^{-1},
h(\lambda) \Lambda h(\lambda)^{-1},
\xi(\cV(\lambda, \theta)))
= (\eta_L e^Q \eta_R^{-1},
\eta_R L(q, p) \eta_R^{-1},
\eta_L \xi(E) \eta_L^{-1})
\label{R_picture_key}
\ee
with some group elements $\eta_L, \eta_R \in K$. 
Now, remembering the definition of the positive 
definite matrix $\cA^\bc$ (\ref{R_Lax_matrix}), 
notice that the $G$-component of the above equation 
immediately leads to the relationship
\be
\cA^\bc(\lambda, \theta)
= h(\lambda)^{-1} \cA(\lambda, \theta) 
h(\lambda)^{-1}
= \eta_R e^{2 Q} \eta_R^{-1},
\ee
from where we get the spectral identification
\be
\sigma(\cA^\bc(\lambda, \theta))
= \{ e^{\pm 2 \check{q}_c(\lambda, \theta)}
\, | \,
c \in \bN_n \}.
\ee  
We see that $\cA^\bc(\lambda, \theta)$ has a 
\emph{simple} spectrum, and the positive 
eigenvalues of the Hermitian matrix 
$\ln(\cA^\bc) / 2$ are exactly the coordinate 
functions $\check{q}_c$ $(c \in \bN_n)$. 

To proceed further, take an arbitrary  
$K \times K$-invariant smooth function 
$f \colon G \rightarrow \bR$ (\ref{f_invariant}),
and consider the naturally generated
reduced Hamiltonian system
$(\cP^R, \omega^R, (\pr_G^* f)^R)$.
Recalling (\ref{R_reduced_Hamiltonian})
and (\ref{R_picture_key}), for the reduced
Hamiltonian we have
\be
(\pr_G^* f)^R(\lambda, \theta)
= f(\cA(\lambda, \theta)^\half h(\lambda)^{-1})
= f(\diag(e^{\check{q}_1(\lambda, \theta)}, 
\ldots, e^{\check{q}_n(\lambda, \theta)},
e^{-\check{q}_1(\lambda, \theta)}, 
\ldots, e^{-\check{q}_n(\lambda, \theta)})),
\ee
i.e. the reduced Hamiltonian depends only on the
coordinates $\check{q}_c$ $(c \in \bN_n)$. Thus, it 
is immediate that the global canonical coordinates
$\check{q}_c$ and $\check{p}_c$ $(c \in \bN_n)$
form an \emph{action-angle system} for the 
mechanical system 
$(\cP^R, \omega^R, (\pr_G^* f)^R)$. Let us 
observe that the action and the angle coordinates 
of the Ruijsenaars picture are coming from the 
pull-backs of the canonical positions and the 
canonical momenta of the Sutherland picture.

Now remember that the reduced Hamiltonian system
$(\cP^R, \omega^R, (\pr_G^* f_1)^R)$ generated
by the $K \times K$-invariant function $f_1$
(\ref{f_1}) coincides with the rational $BC_n$
RSvD model (\ref{H_R}) with three independent 
coupling constants. Therefore the canonical
coordinates $\check{q}_c$ and $\check{p}_c$ 
$(c \in \bN_n)$ provide action-angle variables
for the RSvD model as well. Notice also that the 
positive definite matrix $\cA^\bc$ 
(\ref{R_Lax_matrix}) plays a distinguished role
in the theory of the rational RSvD model. Indeed, 
the positive eigenvalues of the
Hermitian matrix $\ln(\cA^\bc) / 2$ give rise to 
$n$ functionally independent first integrals in 
involution. In particular, the matrix $\cA^\bc$
undergoes an isospectral deformation during the
time evolution of the RSvD dynamics. Also, 
remember that $\cA^\bc$ naturally appears in the 
solution algorithm of the model, as can be seen 
in equation (\ref{RSvD_flows_spec}). Therefore 
$\cA^\bc$ (\ref{R_Lax_matrix}) meets all the
criteria to call it the \emph{Lax matrix} of the 
rational $BC_n$ RSvD model.

To sum up, we constructed action-angle systems 
of canonical coordinates for both the repulsive 
hyperbolic $BC_n$ Sutherland and the rational 
$BC_n$ RSvD models with three independent
coupling constants. The relationships between
the coupling parameters of the corresponding
particle systems are displayed in equations
(\ref{S_coupling_parameters_from_RSvD}) and 
(\ref{RSvD_coupling_parameters_from_S}).
As we have seen, the action 
and the angle coordinates of the Sutherland model 
can be naturally identified with the canonical 
positions and the canonical momenta of the RSvD 
model, and vice versa. That is to say, making use 
of the dual reduction picture (\ref{diagram}), we 
established the \emph{action-angle duality} between 
the standard $BC_n$-type Sutherland and RSvD 
models. This interesting phenomenon was originally 
discovered by Ruijsenaars in the context of the 
$A_n$-type particle systems \cite{RuijCMP1988}. 
Using advanced techniques from symplectic geometry, 
in the last years the $A_n$-type dualities 
have been reinterpreted in the 
reduction framework, too (see the papers 
\cite{FeherKlimcik0906},
\cite{FeherKlimcik_quasi}).
It appears to be an attractive 
research problem for the future
to generalize these techniques to
the non-$A_n$-type setup.

We conclude the paper with some remarks on the 
possible applications of our results.
Besides the natural appearance of the Sutherland 
and the RSvD many-particle systems in the soliton 
scattering description of certain integrable field 
theories (see e.g. \cite{RuijSchneider}, 
\cite{RuijFiniteDimSolitonSystems}
\cite{BabelonBernard}, \cite{KapustinSkorik}), 
we expect that our results find applications in 
the theory of random matrices as well. Indeed, 
by exploiting the existing dualities  
between the $A_n$-type particle systems, the 
authors of the recent papers \cite{Bogomolny_PRL} 
and \cite{Bogomolny} have introduced new classes 
of random matrix ensembles with novel spectral 
statistical properties. Built on the Lax matrices 
$L$ (\ref{L}) and $\cA^\bc$ (\ref{R_Lax_matrix}), 
the proposed dual reduction picture (\ref{diagram}) 
seems to be indispensable in initiating the study 
of the \emph{integrable random matrix ensembles}
associated with non-$A_n$-type root systems.

\medskip
\noindent
\textbf{Acknowledgments.}
This work was partially supported by the Hungarian
Scientific Research Fund (OTKA) under grant
K 77400.



\begin{thebibliography}{99}
\bibitem{RuijSchneider}
S.N.M.~Ruijsenaars, H.~Schneider,
A new class of integrable models and its
relation to solitons,
Ann. Phys. (N.Y.) 170 (1986) 370-405.

\bibitem{RuijFiniteDimSolitonSystems}
S.N.M.~Ruijsenaars,
Finite-dimensional soliton systems,
in: B. Kupershmidt (Ed.), 
Integrable and superintegrable systems,
World Scientific, 1990, pp. 165-206.

\bibitem{BabelonBernard}
O.~Babelon, D.~Bernard,
The sine-Gordon solitons as a $N$-body problem,
Phys. Lett. B 317 (1993) 363-368.

\bibitem{KapustinSkorik}
A.~Kapustin, S.~Skorik,
On the non-relativistic limit of the 
quantum sine-Gordon model with integrable
boundary condition,
Phys. Lett. A 196 (1994) 47-51. 

\bibitem{Bogomolny_PRL}
E.~Bogomolny, O.~Giraud, C.~Schmit,
Random matrix ensembles associated with
Lax matrices,
Phys. Rev. Lett. 103 (2009) 054103.

\bibitem{Bogomolny}
E.~Bogomolny, O.~Giraud, C.~Schmit,
Integrable random matrix ensembles,
manuscript, 
{\tt arXiv:1104.3777}.

\bibitem{OlshaPere76}
M.A.~Olshanetsky, A.M.~Perelomov,
Completely integrable Hamiltonian systems
connected with semisimple Lie algebras,
Invent. Math. 37 (1976) 93-108.

\bibitem{OlshaPere}
M.A.~Olshanetsky, A.M.~Perelomov,
Classical integrable finite-dimensional systems
related to Lie algebras, 
Phys. Rep. 71 (1981) 313-400.

\bibitem{PerelomovBook}
A.M.~Perelomov,
Integrable systems of classical mechanics
and Lie algebras, vol. 1, 
Birkh\"auser, Basel, 1990.

\bibitem{FeherPusztai2007}
L.~Feh\'er, B.G.~Pusztai,
A class of Calogero type reductions of free
motion on a simple Lie group,
Lett. Math. Phys. 79 (2007) 263-277.

\bibitem{vanDiejen1994}
J.F.~van Diejen,
Deformations of Calogero--Moser systems and 
finite Toda chains, 
Theor. Math. Phys. 99 (1994) 549-554.

\bibitem{Pusztai_NPB2011}
B.G.~Pusztai,
Action-angle duality between the $C_n$-type
hyperbolic Sutherland and the rational
Ruijsenaars--Schneider--van Diejen models,
Nucl. Phys. B 853 (2011) 139-173.

\bibitem{Knapp}
A.W.~Knapp,
Lie groups beyond an introduction,
Progress in Mathematics, vol. 140,
Birkh\"auser, Boston, MA, 2002.

\bibitem{AM}
R.~Abraham, J.E.~Marsden,
Foundations of Mechanics, second ed.,
Addison Wesley, 1985.

\bibitem{OR}
J.-P.~Ortega, T.S.~Ratiu,
Momentum maps and Hamiltonian reduction,
Progress in Mathematics, vol. 222,
Birkh\"auser, Boston, MA, 2004.

\bibitem{Pusztai_JPA2011}
B.G.~Pusztai, 
On the scattering theory of the
classical hyperbolic $C_n$ Sutherland model,
J. Phys. A 44 (2011) 155306.

\bibitem{Warner}
F.W.~Warner,
Foundations of differentiable manifolds
and Lie groups,
Springer, New York, 1983. 

\bibitem{RuijCMP1988}
S.N.M.~Ruijsenaars, 
Action-angle maps and scattering theory for 
some finite dimensional integrable systems. 
I. The pure soliton case, 
Commun. Math. Phys. 115 (1988) 127-165.

\bibitem{FeherKlimcik0906}
L.~Feh\'er, C.~Klim\v{c}\'ik,
Poisson--Lie interpretation of trigonometric
Ruijsenaars duality,
Commun. Math. Phys. 301 (2011) 55-104. 

\bibitem{FeherKlimcik_quasi}
L.~Feh\'er, C.~Klim\v{c}\'ik,
Self-duality of the compactified 
Ruijsenaars--Schneider system from 
quasi-Hamiltonian reductions,
manuscript, 
{\tt arXiv:1101.1759}.

\end{thebibliography}
\end{document}